\documentclass[10pt, twocolumn]{IEEEtran}

\usepackage{epsfig,latexsym}
\usepackage{amsmath}
\usepackage{amssymb}
\usepackage{subfigure}
\usepackage{hyperref}
\usepackage{cite}
\usepackage{indentfirst}
\usepackage{times}
\usepackage{fancyhdr}
\usepackage{lastpage}
\usepackage{bigints}
\usepackage{subfigure}
\usepackage{graphicx}
\usepackage{amsthm}
\newtheorem{theorem}{Theorem}

\newtheorem{lemma}{Lemma}
\newtheorem{remark}{Remark}
\newtheorem{definition}{Definition}
\usepackage[noend]{algpseudocode}
\usepackage{algorithmicx,algorithm}
\begin{document}
\title{Asymptotic Performance of the GSVD-Based MIMO-NOMA Communications with Rician Fading}
\author{ Chenguang Rao, Zhiguo Ding, \emph{Fellow, IEEE}, Kanapathippillai Cumanan, \emph{Senior Member, IEEE} and Xuchu Dai\thanks{
		The work is supported by the National Natural Science Foundation of China (No. 61971391) and China Scholarship Council. 
		
		C. Rao and X. Dai are with the CAS Key Laboratory of Wireless-Optical Communications, University of Science and Technology of China, No.96 Jinzhai Road, Hefei, Anhui Province, 230026, P. R. China. (e-mail: rcg1839@mail.ustc.edu.cn; daixc@ustc.edu.cn).
		
		Z. Ding is with Department of Electrical Engineering and Computer Science, Khalifa University, Abu Dhabi, UAE, and Department of Electrical and Electronic Engineering, University of Manchester, Manchester, U.K. (e-mail:, zhiguo.ding@manchester.ac.uk).
		
		Kanapathippillai Cumanan is with the School of Physics, Engineering and Technology, University of York, York YO10 5DD, U.K. (e-mail: kanapathippillai.cumanan@york.ac.uk)
	}\vspace{-2em}}
\maketitle

\begin{abstract} 
	In recent years, the multiple-input multiple-output (MIMO) non-orthogonal multiple-access (NOMA) systems have attracted a significant interest in the relevant research communities. As a potential precoding scheme, the generalized singular value decomposition (GSVD) can be adopted in MIMO-NOMA systems and has been proved to have high spectral efficiency. In this paper, the performance of the GSVD-based MIMO-NOMA communications with Rician fading is studied. In particular, the distribution characteristics of generalized singular values (GSVs) of channel matrices are analyzed. Two novel mathematical tools, the linearization trick and the deterministic equivalent method, which are based on operator-valued free probability theory, are exploited to derive the Cauchy transform of GSVs. An iterative process is proposed to obtain the numerical values of the Cauchy transform of GSVs, which can be exploited to derive the average data rates of the communication system. In addition, the special case when the channel is modeled as Rayleigh fading, i.e., the line-of-sight propagation is trivial, is analyzed. In this case, the closed-form expressions of average rates are derived from the proposed iterative process. Simulation results are provided to validate the derived analytical results.
\end{abstract} 

\vspace{-1em}

\begin{IEEEkeywords} 
	Free deterministic equivalents, generalized singular value decomposition (GSVD), linearization trick, multiply-input multiply-output (MIMO), non-orthogonal multiple access (NOMA), operator-valued free probability, Rician fading channel.
\end{IEEEkeywords} 

\vspace{-1em}

\section{Introduction}
In recent years, with the increasing demand for high-quality and extremely high-throughput communications, multiple-input multiple-output (MIMO) has been considered to be one of the crucial technologies for beyond fifth-generation (B5G) and sixth-generation (6G), and has widely been analyzed and applied in wireless communications \cite{MIMO1,MIMO2,MIMO3,MIMO4}. In many MIMO communication scenarios, more than one user are served by one base station simultaneously. To effectively use the spectrum resources in multi-user MIMO communication scenarios, non-orthogonal multiple access (NOMA) technology has been widely adopted \cite{MIMO-NOMA1,MIMO-NOMA2,MIMO-NOMA3}. With the MIMO-NOMA scheme, the base station serves more than one user in the same time-frequency resource blocks, which has been proved to have superior performance than the traditional MIMO-orthogonal multiple access (OMA) scheme, especially at high signal-to-noise-ratios (SNRs) \cite{SNR}. However, since the base station and users equip multiple antennas in a MIMO-GSVD system, the interference between subchannels cannot be ignored, which motivates the analysis of precoding schemes for MIMO-NOMA \cite{precoding1,precoding2,precoding3}. Among the precoding schemes, the generalized singular value decomposition (GSVD) emerges due to its trade-off for complexity and performance \cite{GSVDo,GSVD1,Chen,OGSVD}. In the GSVD-based MIMO-NOMA communication system, the channel matrices of two users are diagonalized simultaneously thereupon the MIMO channels are converted into several single-input single-output (SISO) subchannels, and there is no interference between each subchannel. The performance of the GSVD-based MIMO-NOMA scheme has been studied in \cite{Chen,GSVDP}, and \cite{OGSVD}. However, all the studies have only analyzed the case of Rayleigh fading, i.e., only non line-of-sight (nLoS) propagation components are considered. In many communication scenarios, the line-of-sight (LoS) propagation cannot be ignored, which demands the channel matrix to be modeled as Rician fading. To the best knowledge of the authors, the performance of GSVD-based MIMO-NOMA system communications with Rician fading has not been studied in the literature, which motivates our work in this paper. 

The main challenge is to find an approach to obtain the distribution characteristic of the generalized singular values (GSVs) of channel matrices. In \cite{Chen}, the GSVs of channel matrices \(\mathbf{H}_1\) and \(\mathbf{H}_2\) are proved to be equal to the eigenvalues of a matrix \(\mathbf{L} = \mathbf{H}_1(\mathbf{H}_2^H\mathbf{H}_2)^{-1}\mathbf{H}_1^H\) when \(\mathbf{H}_2\) has full column rank. In the case of Rayleigh fading, the channel matrix \(\mathbf{H}_i\) is modeled as a Gaussian random matrix, where the random matrix theory can be applied to obtain the probability density function (PDF) of eigenvalues of \(\mathbf{L}\). However, in the case of Rican fading, the channel matrix is modeled as the sum of a Gaussian random matrix and a deterministic matrix, where the random matrix theory cannot be applied directly. Therefore, the operator-valued free probability, which is an extension of free probability, is used in this paper. Two useful tools based on the operator-valued free probability, the linearization trick and free deterministic equivalents, are exploited. In \cite{OFPTRMT}, the linearization trick was proposed and used to convert a complex Gaussian random matrix polynomial problem to a linear addition problem in block random matrices, which can be easily solved by using the subordination theorem. To study more general cases, in \cite{FDEo}, the free deterministic equivalents method was proposed to derive the asymptotic distribution of the eigenvalues of matrix polynomial consisting of self-adjoint and non self-adjoint Gaussian random matrices, deterministic matrices, and Haar-distributed unitary matrices. By exploiting this method, the matrices are replaced with operator-valued random variables, and the matrix polynomial is converted to a new polynomial consisting of operator-valued random variables, whose distribution characteristic can be derived by applying the properties of the operator-valued free probability theorem. The gap between the analytical results derived with the free deterministic equivalents and actual values is proved to vanish when the dimension goes infinite. However, this study dealt with the polynomial whose each element only has a positive degree, while matrix \(\mathbf{L}\) contains the \(-1\) degree, and cannot be applied directly. Motivated by this, we apply the linearization trick and free deterministic equivalents to find a new way to derive the asymptotic distribution characteristic of the eigenvalues of \(\mathbf{L}\), as long as the GSVs of two matrices modeled with Rician fading. The main contributions of this paper are summarized as follows:

\begin{itemize}
	\item A GSVD-based MIMO-NOMA transmission system with Rician fading is considered. For performance analysis, the asymptotic distribution characteristics of the GSVs of channel matrices \(\mathbf{H}_1\) and \(\mathbf{H}_2\) are derived. Two cases under the conditions on number of antennas are discussed respectively. When \(\mathbf{H}_2\) has full column rank, the problem is equivalent to deriving the asymptotic distribution characteristics of the eigenvalues of a matrix \(\mathbf{L} = \mathbf{H}_1(\mathbf{H}_2^H\mathbf{H}_2)^{-1}\mathbf{H}_1^H\). The linearization trick is used twice to convert \(\mathbf{L}\) to a new matrix \(\mathbf{J}\), whose each element is the polynomial of \(\mathbf{H}_i\) with degree \(0\) or \(1\), \(i=1,2\). Then the free deterministic equivalents method is applied to construct an iterative process to obtain the distribution of \(\mathbf{G}\), as long as the GSVs. When \(\mathbf{H}_2\) does not has a full column rank, a full-rank matrix \(\mathbf{H}_3\) with an approximation parameter \(\epsilon\) is constructed, while the pair \(\{\mathbf{H}_1,\mathbf{H}_3\}\) is proved to contains the same GSVs as the pair \(\{\mathbf{H}_1,\mathbf{H}_2\}\) when \(\epsilon\to 0\). Then the distribution of \(\mathbf{G}\) can be derived by using a similar way.
	\item The average rates of two users are derived based on the distribution characteristics of the GSVs of channel matrices. The high-accuracy of these results are verified with the simulation, even if the numbers of antennas are small.
	\item When the channel is modeled as Rayleigh fading, the results derived in the case of Rician fading are simplified, and closed-form expressions of average rates are derived in this special case.
\end{itemize}

\vspace{-2em}

\section{System Model}

Consider a MIMO-NOMA downlink communication system with one base station and two users, which are denoted by \(\mathrm{S}\), \(\mathrm{U}_1\) and \(\mathrm{U}_2\), respectively. \footnote{For the case with more than two users, a hybrid approach can be employed in which users are divided into groups of two users, each group is allocated orthogonal resources, and within each two-user group, the proposed GSVD-based MIMO-NOMA scheme can be utilized.} \(\mathrm{S}\) is equiped with \(N\) antennas, while the \(\mathrm{U}_i\) is equiped with \(M_i\) antennas. The channel matrix between \(\mathrm{S}\) and \(\mathrm{U}_i\) is denoted by \(\mathbf{H}_i\in\mathbb{C}^{M_i\times N}\), which is modeled with Rician fading model. \(\mathbf{H}_i = \bar{\mathbf{H}}_i+\widetilde{\mathbf{H}}_i\), where \(\bar{\mathbf{H}}_i\) is a deterministic matrix that represents the line of sight (LoS) component, and \(\widetilde{\mathbf{H}}_i\) is a random Gaussian matrix with mean \(\mathbf{O}\) and covariance matrix \(\mathbf{I}\) that represents the none-line of sight (nLoS) component. \(\mathbf{I}\) represents an identity matrix, and \(\mathbf{O}\) represents a zero matrix. According to the NOMA scheme, \(\mathrm{S}\) broadcasts the superposed message for two users. Denote \(\mathbf{s}_i\) as the message for the \(\mathrm{U}_i\), \(\mathbf{s}_i\) satisfies \(E\{\mathbf{s}_i\mathbf{s}_i^H\} = \mathbf{I}_N\). The message broadcast by \(\mathrm{S}\) can be expressed as \(\mathbf{x} = \sqrt{l_1}\mathbf{s}_1+\sqrt{l_2}\mathbf{s}_2\). \(0<l_i<1\) represents the power allocation coefficient for \(\mathrm{U}_i\) and satisfies \(l_1+l_2=1\). Denote \(\mathbf{P}\in\mathbb{C}^{N\times N}\) as the precoding matrix, and \(\mathbf{Q}_i\in\mathbb{C}^{M_i\times M_i}\) as the decoding matrix. The received message by \(\mathrm{U}_i\) can be expressed as follows:
\begin{equation}
	\widetilde{\mathbf{y}}_i = \mathbf{Q}_i\mathbf{y}_i = \frac{\sqrt{P}}{t\sqrt{d_i^{\tau}}}\mathbf{Q}_i\mathbf{H}_i\mathbf{P}\mathbf{x}+\mathbf{n},
\end{equation}
where \(d_i\) represents the distance between \(\mathrm{S}\) and \(\mathbf{U}_i\), \(\tau\) represents the large-scale path loss exponent, \(t\) represents the long-term power normalization coefficient, and \(P\) represents the transmission power. \(\mathbf{n}\) is a white noise vector with power \(P_N\). Without loss of generality, the average gain for \(\mathrm{U}_i\) is assumed to be larger than \(\mathrm{U}_2\), i.e., \(\frac{E\{\text{tr}\{\mathbf{H}_1\mathbf{H}_1^H\}\}}{d_1^{\tau}}\geq \frac{E\{\text{tr}\{\mathbf{H}_2\mathbf{H}_2^H\}\}}{d_2^{\tau}}\). To cancel the interference between subchannels, the GSVD precoding scheme is applied. The GSVD of \(\{\mathbf{H}_1,\,\mathbf{H}_2\}\) is denoted as follows:
\begin{equation}
	\mathbf{H}_i = \mathbf{U}_i^H\mathbf{\Sigma}_i\mathbf{V}^{-1},\,i = 1,2,
\end{equation}
where \(\mathbf{U}_i\in\mathbb{C}^{M_i\times M_i}\) is an unitary matrix, \(\mathbf{V}\in\mathbb{C}^{N\times N}\) is an invertible matrix, and \(\mathbf{\Sigma}_i\in \mathbb{C}^{M_i\times N}\) is a rectangular diagonal matrix with \(S\) non-zero diagonal elements. From \cite{GSVDo}, \(S\) can be known as \(S = \min\{M_1,N\}+\min\{M_2,N\}-\min\{M_1+M_2,N\}\). Denote \(\mathbf{\Sigma}_i = \text{diag}\{\sigma_{i,1},\sigma_{i,2},...,\sigma_{i,S}\}\). The GSV \(\omega_j,j = 1,2,...,S\) is defined as follows:
\begin{equation}
	\omega_j = \frac{\sigma_{1,j}^2}{\sigma_{2,j}^2}.
\end{equation}
Then \(\sigma_{i,j}\) can be expressed by \(\omega_j\) as \(\sigma_{1,j}^2 = \frac{\omega_j}{1+\omega_j}\) and \(\sigma_{2,j}^2 = \frac{1}{1+\omega_j}\). Set the precoding and decoding matrices as \(\mathbf{P} = \mathbf{V}\) and \(\mathbf{Q}_i = \mathbf{U}_i\), then the received message \(\widetilde{\mathbf{y}}_i\) can be expressed as follows:
\begin{equation}
	\begin{aligned}
		\widetilde{\mathbf{y}}_i =  \frac{\sqrt{P}}{t\sqrt{d_i^{\tau}}}\mathbf{U}_i\mathbf{H}_i\mathbf{V}\mathbf{x}+\mathbf{n}
		=\frac{\sqrt{P}}{t\sqrt{d_i^{\tau}}}\mathbf{\Sigma}_i\mathbf{x}+\mathbf{n},
	\end{aligned}
\end{equation}
and the message at \(j\)-th
subchannel of \(\widetilde{\mathbf{y}}_i\) can be expressed as follows:
\begin{equation}
	y_{i,j} = \frac{\sqrt{P}}{t\sqrt{d_i^{\tau}}}\sigma_{i,j}x_j+n_j,
\end{equation}
where \(x_j = \sqrt{l_1}s_{1,j}+\sqrt{l_2}s_{2,j}\) represents the \(j\)-th element of \(\mathrm{U}_i\). The MIMO channel is now converted into several parallel SISO subchannels, where the successive interference cancellation (SIC) can be applied to eliminate interference. In this paper, the statistical channel state information (CSI)-based SIC is used \cite{sCSI}, i.e., the SIC is applied in the user who has a larger average power of channel fading. Since \(\frac{E\{\text{tr}\{\mathbf{H}_1\mathbf{H}_1^H\}\}}{d_1^{\tau}}\geq \frac{E\{\text{tr}\{\mathbf{H}_2\mathbf{H}_2^H\}\}}{d_2^{\tau}}\), SIC is applied in \(\mathrm{U}_1\). Specifically, \(\mathrm{U}_1\) decodes \(s_2\) first, and then decoding \(s_1\) after eliminating \(s_2\). \(\mathrm{U}_2\) decodes \(s_2\) by regarding \(s_1\) as noise directly. Then the rate of \(s_{i,j}\) can be expressed as follows:
\begin{equation}
	R_{1,j} = \log\left(1 + \frac{\omega_j}{1+\omega_j}\frac{l_1\rho}{td_1^{\tau}}\right),
\end{equation}
\begin{equation}
	R_{2,j} = \log\left(1 + \frac{l_2}{l_1+(1+\omega_j)\rho^{-1}td_2^{\tau}}\right),
\end{equation}
where \(\rho = \frac{P}{P_N}\) SNR. In this paper, the average performance \(E\{R_{i,j}\}\) for each sub-channel is considered, i.e., \(\omega_j\) is regarded as the same with any \(j\). Thus, the average rate of \(\mathrm{U}_i\) can be expressed as follows:
\begin{equation}\label{aR1}
	\bar{R}_1 = S\int_{0}^{+\infty}\log\left(1 + \frac{x }{1+x}\frac{l_1\rho}{td_1^{\tau}}\right)dF_{\omega}(x),
\end{equation}
\begin{equation}\label{aR2}
	\bar{R}_2 = S\int_{0}^{+\infty}\log\left(1 + \frac{l_2}{l_1+(1+x)\rho^{-1}td_2^{\tau}}\right)dF_{\omega}(x),
\end{equation}
where \(F_{\omega}(x)\) denotes the cumulative density function (CDF) of \(\omega\).

\section{Performance Analysis}
In this work, the numbers of antennas are assumed to satisfy \(M_1 \leq M_2\) and \(M_1+M_2> N\). Since the GSVs of \(\{\mathbf{H}_1,\mathbf{H}_2\}\) are reciprocal of the GSVs of \(\{\mathbf{H}_2,\mathbf{H}_1\}\), the case \(M_1 > M_2\) is equivalent to the case \(M_1 \leq M_2\) by swapping \(\mathbf{H}_1\) and \(\mathbf{H}_2\). As for the case of \(M_1+M_2 \leq N\), it can be understood from \cite{Chen} that \(\omega_j\) is deterministic, and has little value for analysis.

From \cite{Chen}, it can be known that when \(M_2\geq N\), \(\omega_j\) equals to the non-zero eigenvalues of a matrix \(\mathbf{L}\), which is expressed as follows:

\vspace{-1em}

\begin{equation}
	\mathbf{L} = \mathbf{H}_1(\mathbf{H}_2^H\mathbf{H}_2)^{-1}\mathbf{H}_1^H.
\end{equation}

\vspace{-1em}

Then we focus on the approach to derive the eigenvalue distribution of \(\mathbf{L}\). Before the analysis, the Cauchy transform is used, which contains the asymptotic distribution characteristics of random variables \cite{Cauchy}.

\vspace{-1em}

\begin{definition}
	\(\mu\) is a random variable with \(F_{\mu}(x)\) CDF, then its Cauchy transform of \(\mu\) is defined as follows:
	\begin{equation}
		G_{\mu}(z) = \int_{R}\frac{1}{z-x}dF_{\mu}(x).
	\end{equation}
\end{definition}
In this paper, \(z\) is restricted as \(\text{Re}(z)<0\), where \(\text{Re}\{z\}\) represents the real part of \(z\). The problem is now transformed to obtaining \(G_{\omega}(z)\). Before the analysis, the operator-valued free probability is introduced as the mathematical preliminary in the following subsection.

\vspace{-1em}

\subsection{Free Probability and Operator-Valued Free Probability}
Free probability theory is first proposed by Voiculescu in 1985 as a tool to analyze non-commutative random variables, for example, matrices. Operator-valued free probability is an extension of free probability, which is first proposed by Voiculescu in \cite{OFPTo} to solve the matrix polynomial problems. In this subsection, the fundamental concepts of free probability and operator-valued free probability are briefly introduced, respectively. More details of free probability and operator-valued free probability can be found in the Appendix A in \cite{Lu} and \cite{FPT}.

The definition of the non-commutative probability space is presented as follows:

\begin{definition}
	Denote \(\mathcal{A}\) as a unital non-commutative algebra over \(\mathbb{C}\) with a unit \(1_{\mathcal{A}}\) and a unital linear functional \(\phi:\mathcal{A}\to\mathbf{C}\) satisfying \(\phi(1_{\mathcal{A}}) = 1\). Then the pair \((\mathcal{A},\phi)\) is defined as a non-commutative probability space. The elements \(a\in \mathcal{A}\) are called as non-commutative random variables.
\end{definition}

Based on the definition of non-commutative probability space, the operator-valued non-commutative probability space is defined as follows:

\begin{definition}
	Denote \(\mathcal{A}\) as a unital non-commutative algebra over \(\mathbb{C}\) with a unit \(1_{\mathcal{A}}\), and denote \(\mathcal{B}\subset\mathcal{A}\) as a unital subalgebra of \(\mathcal{A}\). Define a unital linear functional \(\phi:\mathcal{A}\to\mathcal{B}\) satisfying:
	\begin{itemize}
		\item \(\phi(B) = B\), for any \(B\in\mathcal{B}\).
		\item \(\phi(B_1AB_2) = B_1\phi(A)B_2\), for any \(A\in\mathcal{A}\) and \(B_1,B_2\in\mathcal{B}\).
	\end{itemize}
	Then the triplet \((\mathcal{A},\phi,\mathcal{B})\) is defined as a \(\mathcal{B}\)-valued non-commutative probability space.
\end{definition}

In this paper, free probability theory and operator-valued free probability theory are applied in random matrix theory. Denote \(\mathcal{M}_n\) as the algebra of \(n\times n\) complex random matrices. For a element \(\mathbf{X}\in\mathcal{M}_n\), define \(\phi(\mathbf{X}) = \frac{1}{n}E\{\text{Tr}\{\mathbf{X}\}\}\) and \(1_{\mathcal{M}_n} = \mathbf{I}_n\). Then \((\mathcal{M}_n,\phi)\) is a non-commutative probability space. 

Denote \(\mathcal{D}_n\) as the algebra of \(n\times n\) complex diagonal matrices. Define a map \(E_{\mathcal{D}_n}:\mathcal{M}_n\to\mathcal{D}_n\) defined as follows:

\begin{equation}
	E_{\mathcal{D}_n}\{\mathbf{X}\} = \text{diag}\{E\{X_{1,1}\},E\{X_{2,2}\},...,E\{X_{n,n}\}\}.
\end{equation}
The triplet \((\mathcal{M}_n,E_{\mathcal{D}_n},\mathcal{D}_n)\) is a \(\mathcal{D}_n\)-valued probability space. Then the Cauchy transform can be extended over \((\mathcal{M}_n,E_{\mathcal{D}_n})\) as the \(\mathcal{D}_n\)-valued Cauchy transform, which is defined as follows:

\begin{definition}
	\((\mathcal{M}_n,E_{\mathcal{D}_n},\mathcal{D}_n)\) is a \(\mathcal{D}_n\)-valued probability space. Let \(\mathbf{X}\in\mathcal{M}_n\), then \(\mathcal{D}_n\)-valued Cauchy transform of \(\mathbf{X}\) is defined as follows:
	\begin{equation}
		\mathcal{G}_{\mathbf{X}}^{\mathcal{D}_n}(\mathbf{A}) = E_{\mathcal{D}_n}\{\left(\mathbf{A-X}\right)^{-1}\},
	\end{equation}
	where \(\mathbf{A}\in\mathcal{D}_n\) satisfies \(\frac{1}{2i}(\mathbf{A}-\mathbf{A}^H)\succ 0\).
\end{definition}

Assume the expected cumulative distribution of the eigenvalues of \(\mathbf{X}\) converges to a random variable \(x\) when \(n\to+\infty\), it can be known that \(G_{x}(z) = \frac{1}{n}\text{Tr}\{\mathcal{G}_{\mathbf{X}}^{\mathcal{D}_n}(z\mathbf{\mathbf{I}_n})\}\) \cite{Lu}. Thus, \(G_{\omega}(z)\) can be obtained by deriving \(\mathcal{G}_{\mathbf{L}}^{\mathcal{D}_n}(z\mathbf{\mathbf{I}_n})\). However, \(\mathcal{G}_{\mathbf{L}}^{\mathcal{D}_n}(z\mathbf{\mathbf{I}_n})\) is still complex to derive. Thus, another method, the linearization trick, is exploited, which will be introduced in the following subsection.

\subsection{Linearization Trick}

In \cite{LTP,OFPTRMT}, the linearization trick is introduced, which is required in our work. The  main theorem of the linearization trick is introduced as follows:
\begin{theorem}\label{LTP}
	Assume \(\mathbf{X}_i,i = 1,2,...,M\) are complex random matrices in the \(\mathcal{D}_n\)-valued probability space \((\mathcal{M}_n,E^{\mathcal{D}_n})\). Assume \(\mathbf{P}\) is the polynomial of \(\mathbf{X}_i\) and can be expressed as follows:
	\begin{equation}
		\mathbf{P} = -\mathbf{u}^H\mathbf{Q}^{-1}\mathbf{v},
	\end{equation}
	where \(\mathbf{u}, \mathbf{v}\in\mathbb{C}^{K-1\times 1}\), \(\mathbf{Q}\in \mathbb{C}^{K-1\times K-1}\) is invertible. Each element in \(\mathbf{u}, \mathbf{v}\) and \(\mathbf{Q}\) is the polynomial of \(\mathbf{X}_i\) with degree \(\leq 1\). Then a new matrix over \(\mathcal{D}_{K}\)-valued probability space can be constructed as follows:
	\begin{equation}
		\hat{\mathbf{P}} = 
		\begin{pmatrix}
			0&\mathbf{u}^H\\
			\mathbf{v}&\mathbf{Q}
		\end{pmatrix}.
	\end{equation}
	Define a matrix
	\begin{equation}
		\mathbf{B} = 
		\begin{pmatrix}
			z\mathbf{I}_n &\mathbf{O}\\
			\mathbf{O}&\mathbf{O}_{K-n}
		\end{pmatrix}.
	\end{equation}
	Then the \(\mathcal{D}_K\)-valued Cauthy transform of \(\hat{\mathbf{P}}\) can be obtained as
	\begin{equation}\label{LT}
		\mathcal{G}^{\mathcal{D}_{K}}_{\hat{\mathbf{P}}}(\mathbf{B}) = 
		\begin{pmatrix}
			\mathcal{G}^{\mathcal{D}_{n}}_{\mathbf{P}}(z\mathbf{I}_n)&*&\\
			*&*
		\end{pmatrix}.
	\end{equation}
\end{theorem}
\begin{proof}
	See \cite{OFPTRMT} and \cite{LTP}. 
\end{proof}

In this theorem, \(\hat{\mathbf{P}}\) is called as the linearization of \(\mathbf{P}\).  \(\mathcal{G}^{\mathcal{D}_{n}}_{\mathbf{P}}(z\mathbf{I}_n)\) can be obtained by deriving \(\mathcal{G}^{\mathcal{D}_{K}}_{\hat{\mathbf{P}}}(\mathbf{B})\). Since \(\hat{\mathbf{P}}\) can be considered as the sum of selfadjoint operator-valued matrix polynomials of \(\mathbf{X}_i,i = 1,2,...,M\), \(\mathcal{G}^{\mathcal{D}_{K}}_{\hat{\mathbf{P}}}(\mathbf{B})\) can be derived by using the free deterministic equivalents and subordination formulation.

\begin{remark}
	There is a condition in the linearization trick that the polynomial entries in \(\mathbf{u}, \mathbf{v}\) and \(\mathbf{Q}\) must have degree \(\leq 1\). However, (\ref{LT}) also holds without this condition. In this paper, we will use this theorem without the condition.
\end{remark}

\subsection{The Approach to Obtain \(G_{\omega}(z)\)}\label{App}
In this subsection, the cases of \(M_2 \geq N\) and \(M_2 < N < M_1+M_2\) are discussed. 

\subsubsection{The case \(M_2 \geq N\)}

When \(M_2 \geq N\), \(\omega\) equals to the non-zero eigenvalues of \(\mathbf{L} = \mathbf{H}_1(\mathbf{H}_2^H\mathbf{H}_2)^{-1}\mathbf{H}_1^H\). Based on this precondition, the main process to obtain \(G_{\omega}(z)\) can be divided into several steps as follows: 
\begin{itemize}
	\item Apply the linearization trick to construct a matrix \(\mathbf{J}\) that \(\mathcal{G}^{\mathcal{D}_{M_1}}_{\mathbf{L}}(z\mathbf{I}_n)\) is in the \(\mathcal{G}^{\mathcal{D}_{4n}}_{\mathbf{J}}(z\mathbf{I}_n)\), while each element in \(\mathbf{J}\) is the polynomial of \(\mathbf{H}_i\) with degree \(0\) or \(1\).
	\item Apply the free deterministic equivalents method \cite{Lu} to construct a matrix \(\boldsymbol{\mathcal{J}}\) by satisfying \(\lim\limits_{n\to +\infty}\mathcal{G}^{\mathcal{D}_{4n}}_{\mathbf{J}}(z\mathbf{I}_n)-\mathcal{G}^{\mathcal{D}_{4n}}_{\boldsymbol{\mathcal{J}}}(z\mathbf{I}_n) = 0\).
	\item Derive \(\mathcal{G}^{\mathcal{D}_{4n}}_{\boldsymbol{\mathcal{J}}}(z\mathbf{I}_n)\) with subordination theorem. Then propose an iterative approach to obtain the approximated results of \(\mathcal{G}^{\mathcal{D}_{4n}}_{\mathbf{J}}(z_0\mathbf{I}_n)\) and \(\mathcal{G}^{\mathcal{D}_{M_1}}_{\mathbf{L}}(z_0\mathbf{I}_n)\).
	\item Find the relationship between \(G_{\omega}(z)\) and \(\mathcal{G}^{\mathcal{D}_{M_1}}_{\mathbf{L}}(z\mathbf{I}_n)\) to obtain \(G_{\omega}(z_0)\).
\end{itemize}

The detailed process is presented as follows. First, by applying Theorem \ref{LTP}, a matrix \(\hat{\mathbf{L}}\) can be constructed as follows:
\begin{equation}
	\hat{\mathbf{L}} =
	\begin{pmatrix}
		\mathbf{O} &\mathbf{H}_1 \\
		\mathbf{H}_1^H &-\mathbf{H}_2^H\mathbf{H}_2
	\end{pmatrix}.
\end{equation}
Denote \(n = M_1+N\), then 
\begin{equation}
	\mathcal{G}^{\mathcal{D}_{n}}_{\mathbf{\hat{\mathbf{L}}}}(\mathbf{B}) =
	\begin{pmatrix}
		\mathcal{G}^{\mathcal{D}_{M_1}}_{\mathbf{L}}(z\mathbf{I}_{M_1})&\mathbf{O}\\
		\mathbf{O}&*
	\end{pmatrix},
\end{equation}
where
\begin{equation}
	\mathbf{B} = 
	\begin{pmatrix}
		z\mathbf{I}_{M_1} &\mathbf{O}\\
		\mathbf{O}&\mathbf{O}_{N}
	\end{pmatrix}.
\end{equation}

Now the target is to obtain \(\mathcal{G}^{\mathcal{D}_{n}}_{\hat{\mathbf{L}}}(\mathbf{B})\). However, the one of the elements in \(\hat{\mathbf{L}}\) is the \(2\)-degree of \(\mathbf{H}_2\), which makes the problem challenging to solve. Therefore, it is necessary to apply Theorem \ref{LTP} to \(\hat{\mathbf{L}}\) again. Denote \(\mathbf{X}_i \in \mathbb{C}^{n\times n}\) as follows:
\begin{equation}
	\mathbf{X}_1 = 
	\begin{pmatrix}
		\mathbf{O} &\mathbf{H}_1\\
		\mathbf{H}_1^H &\mathbf{O}
	\end{pmatrix},
	\mathbf{X}_2 = 
	\begin{pmatrix}
		\mathbf{O}_{(n-M_2)\times M_1} &\mathbf{O}\\
		\mathbf{O} &\mathbf{H}_2
	\end{pmatrix}.
\end{equation}
Then \(\hat{\mathbf{L}}\) can be expressed as \(\hat{\mathbf{L}} = \mathbf{X}_1-\mathbf{X}_2^H\mathbf{X}_2\). Now rewrite \(\hat{\mathbf{L}}\) as
\begin{equation}
	\hat{\mathbf{L}} = -
	\begin{pmatrix}
		\mathbf{X}_2^H&\mathbf{O}&\mathbf{I}_{n}
	\end{pmatrix}
	\begin{pmatrix}
		\mathbf{I}_{n} &\mathbf{O} &\mathbf{O}\\
		\mathbf{O} &\mathbf{X}_1&-\mathbf{I}_{n}\\
		\mathbf{O} &-\mathbf{I}_{n}&\mathbf{O}
	\end{pmatrix}^{-1}
	\begin{pmatrix}
		\mathbf{X}_2\\\mathbf{O}\\\mathbf{I}_{n}
	\end{pmatrix}.
\end{equation}
By applying Theorem \ref{LTP}, \(\mathbf{J}\) can be constructed as follows:
\begin{equation}
	\mathbf{J} = 
	\begin{pmatrix}
		\mathbf{O} &\mathbf{X}_2^H&\mathbf{O}&\mathbf{I}_{n}\\
		\mathbf{X}_2 &\mathbf{I}_{n} &\mathbf{O} &\mathbf{O}\\
		\mathbf{O} &\mathbf{O} &\mathbf{X}_1&-\mathbf{I}_{n}\\
		\mathbf{I}_{n} &\mathbf{O} &-\mathbf{I}_{n}&\mathbf{O}
	\end{pmatrix}.
\end{equation}
Let
\begin{equation}
	\mathbf{C} = 
	\begin{pmatrix}
		\mathbf{B} &\mathbf{O}\\
		\mathbf{O}&\mathbf{O}_{3n}
	\end{pmatrix},
\end{equation}
then \(\mathcal{G}^{\mathcal{D}_{M_1}}_{\hat{\mathbf{L}}}(\mathbf{z\mathbf{I}_{M_1}})\) can be derived from:
\begin{equation}
	\mathcal{G}^{\mathcal{D}_{4n}}_{\mathbf{J}}(\mathbf{C}) = 
	\begin{pmatrix}
		\mathcal{G}^{\mathcal{D}_{n}}_{\hat{\mathbf{L}}}(\mathbf{B})&\mathbf{O}\\
		\mathbf{O}&*
	\end{pmatrix}
	=
	\begin{pmatrix}
		\mathcal{G}^{\mathcal{D}_{M_1}}_{\mathbf{L}}(\mathbf{z\mathbf{I}_{M_1}})&\mathbf{O}\\
		\mathbf{O}&*
	\end{pmatrix}.
\end{equation}

Note that \(\mathbf{J}\) is consisted of the \(1\)-degree of \(\mathbf{H}_1\) and \(\mathbf{H}_2\), it can be expressed as the sum of a deterministic matrix and a random hermitian matrix as \(\mathbf{J} = \bar{\mathbf{J}}+\widetilde{\mathbf{J}}\). The detailed matrices are expressed as follows:
\begin{equation}
	\begin{aligned}
		\begin{small}
			\bar{\mathbf{J}} = \begin{pmatrix}
				\mathbf{O} &\bar{\mathbf{X}}_2^H&\mathbf{O}&\mathbf{I}_{n}\\
				\bar{\mathbf{X}}_2 &\mathbf{I}_{n} &\mathbf{O} &\mathbf{O}\\
				\mathbf{O} &\mathbf{O} &\bar{\mathbf{X}}_1&-\mathbf{I}_{n}\\
				\mathbf{I}_{n} &\mathbf{O} &-\mathbf{I}_{n}&\mathbf{O}
			\end{pmatrix},
			\widetilde{\mathbf{J}} = \begin{pmatrix}
				\mathbf{O} &\widetilde{\mathbf{X}}_2^H&\mathbf{O}&\mathbf{O}\\
				\widetilde{\mathbf{X}}_2 &\mathbf{O} &\mathbf{O} &\mathbf{O}\\
				\mathbf{O} &\mathbf{O} &\widetilde{\mathbf{X}}_1&\mathbf{O}\\
				\mathbf{O} &\mathbf{O} &\mathbf{O}&\mathbf{O}
			\end{pmatrix},
		\end{small}
	\end{aligned}
\end{equation}
\begin{equation}
	\bar{\mathbf{X}}_1 = 
	\begin{pmatrix}
		\mathbf{O}&\bar{\mathbf{H}}_1\\
		\bar{\mathbf{H}}_1^H&\mathbf{O}
	\end{pmatrix},
	\widetilde{\mathbf{X}}_1 = 
	\begin{pmatrix}
		\mathbf{O}&\widetilde{\mathbf{H}}_1\\
		\widetilde{\mathbf{H}}_1^H&\mathbf{O}
	\end{pmatrix},
\end{equation}
\begin{equation}\label{X2}
	\bar{\mathbf{X}}_2 = 
	\begin{pmatrix}
		\mathbf{O}_{(n-M_2)\times M_1} &\mathbf{O}\\
		\mathbf{O} &\bar{\mathbf{H}}_2
	\end{pmatrix},
	\widetilde{\mathbf{X}}_2 = 
	\begin{pmatrix}
		\mathbf{O}_{(n-M_2)\times M_1} &\mathbf{O}\\
		\mathbf{O} &\widetilde{\mathbf{H}}_2
	\end{pmatrix}.
\end{equation}

Then the free deterministic equivalents method is applied, which is introduced in \cite{FPT,Lu}. Since each element in \(\widetilde{\mathbf{J}}\) is either zero or a Gaussian random variable with unit variance, and always zero in the diagonal, it can be replaced by a matrix consisting of freely independent centered circular elements with unit variance. This is the construction of the free deterministic equivalent of \(\mathbf{J}\). The detailed process is presented as follows. Denote \(\mathcal{A}\) as a unital algebra and \((\mathcal{A},\phi)\) as a non-commutative probability space. Define \(\widetilde{\boldsymbol{\mathcal{H}}}_i\in\mathcal{A}^{M_1\times N}\), whose entries are freely independent centered circular elements with unit variances \cite{FPT}. Set \(\boldsymbol{\mathcal{H}}_i = \bar{\mathbf{H}}_i + \widetilde{\boldsymbol{\mathcal{H}}}_i\). Then by replacing all \(\mathbf{H}_i\) with \(\boldsymbol{\mathcal{H}}_i\) in \(\mathbf{J}\), \(\boldsymbol{\mathcal{J}}\) can be constructed as the free deterministic equivalent of \(\mathbf{J}\), and \(\boldsymbol{\mathcal{J}}\) can be written as \(\boldsymbol{\mathcal{J}} = \bar{\mathbf{J}} + \widetilde{\boldsymbol{\mathcal{J}}}\), where \(\widetilde{\boldsymbol{\mathcal{J}}}\) can be constructed by replacing \(\widetilde{{\mathbf{H}}}_i\) with \(\widetilde{\boldsymbol{\mathcal{H}}}_i\) in \(\widetilde{\mathbf{J}}\). Denote \(\mathcal{M}_n(\mathcal{A})\) as the algebra of \(n\times n\) random matrices consisted with the elements in \(\mathcal{A}\). With an element \(\boldsymbol{\mathcal{X}}\in\mathcal{M}_n(\mathcal{A})\), define a map \(\mathcal{E}_{\mathcal{D}_n}:\mathcal{M}_n(\mathcal{A})\to\mathcal{D}_n\) as follows:

\begin{equation}
	\mathcal{E}_{\mathcal{D}_n}\{\boldsymbol{\mathcal{X}}\} = \text{diag}\{\phi(\mathcal{X}_{1,1}),\phi(\mathcal{X}_{2,2}),...,\phi(\mathcal{X}_{n,n})\}.
\end{equation}
It can be inferred that for any deterministic matrix \(\mathbf{Y}\in\mathcal{M}_n\), \(\mathcal{E}_{\mathcal{D}_{N}}\{\widetilde{\boldsymbol{\mathcal{H}}}_i^H\mathbf{Y}\widetilde{\boldsymbol{\mathcal{H}}}_i\} = E_{\mathcal{D}_N}\{\widetilde{\mathbf{H}}_i^H\mathbf{Y}\widetilde{\mathbf{H}}_i\}\).
Define the \(\mathcal{D}_n\)-valued Cauchy transform of \(\mathbf{X}\) as follows:
\begin{equation}
	\mathcal{G}_{\boldsymbol{\mathcal{X}}}^{\mathcal{D}_n}(\mathbf{A}) = \mathcal{E}_{\mathcal{D}_n}\{\left(\mathbf{A-\boldsymbol{\mathcal{X}}}\right)^{-1}\},
\end{equation}
where \(\mathbf{A}\in\mathcal{D}_n\) satisfies \(\frac{1}{2i}(\mathbf{A}-\mathbf{A}^H)\succ 0\). Then the following theorem is presented:

\begin{theorem}\label{FDT}
	\(\boldsymbol{\mathcal{J}}\) and \(\mathbf{J}\) satisfying:
	\begin{equation}\label{JF}
		\lim\limits_{n\to +\infty}\mathcal{G}^{\mathcal{D}_{4n}}_{\mathbf{J}}(z\mathbf{I}_n)-\mathcal{G}^{\mathcal{D}_{4n}}_{\boldsymbol{\mathcal{J}}}(z\mathbf{I}_n) = 0.
	\end{equation}
\end{theorem}
\begin{proof}
	See Appendix A-B in \cite{Lu}, and this theorem is a special case of the one in \cite{Lu}. 
\end{proof}

Then by deriving \(\mathcal{G}^{\mathcal{D}_{4n}}_{\boldsymbol{\mathcal{J}}}(z\mathbf{I}_n)\), \(\mathcal{G}^{\mathcal{D}_{4n}}_{\mathbf{J}}(\mathbf{C})\) can be obtained. Before that, a Lemma is required and introduced:
\begin{lemma}\label{ED}
	Assume \(\mathbf{X}\in\mathbb{C}^{m\times n}\) is a Gaussian random matrix with mean \(\mathbf{O}\) and covariance matrix \(\mathbf{I}\), \(\mathbf{Y}\in\mathbb{C}^{n\times n}\) is a deterministic diagonal matrix independent from \(\mathbf{X}\), then 
	\begin{equation}
		E_{\mathcal{D}_{m}}\{\mathbf{X}\mathbf{Y}\mathbf{X}^H\} = \text{Tr}\{\mathbf{Y}\}\mathbf{I}_m.
	\end{equation}
\end{lemma}
\begin{proof}
	See Appendix \ref{EDA}.
\end{proof}

Then the approach to derive \(\mathcal{G}^{\mathcal{D}_{4n}}_{\boldsymbol{\mathcal{J}}}(z\mathbf{I}_n)\)  can be shown as follows:

\begin{theorem}\label{IT}
	Divide \(\mathcal{G}^{\mathcal{D}_{4n}}_{\boldsymbol{\mathcal{J}}}(\mathbf{C})\) into several block matrices as follows:
	
	Denote \(\mathcal{G}^{\mathcal{D}_{4n}}_{\boldsymbol{\mathcal{J}}}(\mathbf{C}) = \text{blkdiag}\{\mathbf{E}_1,\mathbf{E}_2,\mathbf{E}_3,\mathbf{E}_4\}\), and \(\mathbf{E}_i = \text{blkdiag}\{\mathbf{E}_{i,1},\mathbf{E}_{i,2}\}\), where \(\mathbf{E}_{1,1}\in\mathbb{C}^{M_1\times M_1}\), \(\mathbf{E}_{1,2}\in\mathbb{C}^{N\times N}\), \(\mathbf{E}_{2,1}\in\mathbb{C}^{n-M_2\times n-M_2}\), \(\mathbf{E}_{2,2}\in\mathbb{C}^{M_2\times M_2}\), \(\mathbf{E}_{3,1}\in\mathbb{C}^{M_1\times M_1}\), and \(\mathbf{E}_{3,2}\in\mathbb{C}^{N\times N}\) are diagonal matrices. Then \(\mathbf{E}_{i,j}\) satisfy the following equations:
	\begin{equation}\label{E11}
		\begin{aligned}
			\mathbf{E}_{1,1}  =  \text{diag}\{((z-\text{Tr}\{\mathbf{E}_{1,2}\})\mathbf{I}_{M_1}
			&-\bar{\mathbf{H}}_1\mathbf{A}_1^{-1}\bar{\mathbf{H}}_1^H)^{-1}\},
		\end{aligned}
	\end{equation}
	\begin{equation}\label{E12}
		\begin{aligned}
			\mathbf{E}_{1,2} =  \text{diag}\{(\mathbf{A}_1-(z-\text{Tr}\{\mathbf{E}_{1,2}\})^{-1}\bar{\mathbf{H}}_1^H\bar{\mathbf{H}}_1)^{-1}\},
		\end{aligned}
	\end{equation}
	\begin{equation}\label{E22}
		\begin{aligned}
			\mathbf{E}_{2,2} = \text{diag}
			\{
			(\bar{\mathbf{H}}_2\mathbf{A}_2^{-1}\bar{\mathbf{H}}_2^H-(1+\text{Tr}\{\mathbf{E}_{1,2}\})\mathbf{I}_{M_2})^{-1}
			\},
		\end{aligned}
	\end{equation}
	
	where
	\begin{equation}\label{A1}
		\mathbf{A}_{1} = (1+\text{Tr}\{\mathbf{E}_{1,2}\})^{-1}\bar{\mathbf{H}}_2^H\bar{\mathbf{H}}_2-(\text{Tr}\{\mathbf{E}_{2,2}\}+\text{Tr}\{\mathbf{E}_{1,1}\})\mathbf{I}_N,
	\end{equation}
	\begin{equation}\label{A2}
		\mathbf{A}_{2} = (z-\text{Tr}\{\mathbf{E}_{1,2}\})^{-1}\bar{\mathbf{H}}_1^H\bar{\mathbf{H}}_1+(\text{Tr}\{\mathbf{E}_{2,2}\}+\text{Tr}\{\mathbf{E}_{1,1}\})\mathbf{I}_N.
	\end{equation}
\end{theorem}
\begin{proof}
	See Appendix \ref{ITA}.
\end{proof}

From the above theorem, the numerical results of \(\mathbf{E}_{i,j}\) with a certain input \(z_0\) can be obtained with iteration. Then the asymptotic result of \(G_{\mathbf{L}}(z_0)\) can be obtained by
\begin{equation}\label{GL}
	G_{\mathbf{L}}(z_0) \approx \frac{1}{M_1}\text{Tr}\{\mathbf{E}_{1,1}\}.
\end{equation}
The accuracy improves when \(n\) increases, i.e., when the number of antennas increases. Therefore, the approach provides more accurate results in large-scale MIMO scenarios. 

To obtain \(G_{\omega}(z_0)\), the following theorem is introduced:
\begin{theorem}\label{Gw}
	\(G_{\omega}(z_0)\) can be expressed as follows:
	\begin{equation}\label{GwE}
		G_{\omega}(z_0) = \frac{M_1}{S}G_{\mathbf{L}}(z_0)-\frac{M_1-S}{Sz_0}.
	\end{equation}
\end{theorem}
\begin{proof}
	See Appendix \ref{GwA}.
\end{proof}

\subsubsection{The case \(M_2 < N < M_1 + M_2\)}

When \(M_2 < N <M_1+M_2\), it is difficult to find a matrix that is constructed with \(\mathbf{H}_1\) and \(\mathbf{H}_2\) which has the same eigenvalues as \(\omega\). Therefore, we add rows to the bottoms of \(\mathbf{H}_2\) to transfer this case to the \(M_2 = N\) case. Before that, an important theorem is introduced:
\begin{theorem}\label{ellipse}
	For two matrices \(\mathbf{H}_1\in\mathbb{C}^{M_1\times N}\), \(\mathbf{H}_2\in\mathbb{C}^{M_2\times N}\), \(M_2 < N < M_1 + M_2\). \(S = M_1+M_2-N\). Construct a new matrix
	\begin{equation}\label{FR}
		\mathbf{H}_3 = 
		\begin{pmatrix}
			\mathbf{H}_2\\
			\epsilon\mathbf{F}
		\end{pmatrix},
	\end{equation}
	where \(\mathbf{F}\in\mathbb{C}^{N-M_2\times N}\) and satisfies \(\text{rank}\{\mathbf{H}_3\} = N\). \(\epsilon>0\) is an approximation parameter. Assume the GSVs of pairs \(\{\mathbf{H}_1,\mathbf{H}_3\}\) is sorted as from the smallest to the largest \(\mu_1<\mu_2<...<\mu_{M_1}\). Then when \(\epsilon\to 0\), \(\mu_1,\mu_2,...,\mu_S\) continuously converges to the GSVs of pairs \(\{\mathbf{H}_1,\mathbf{H}_2\}\), and \(\mu_{S+1},\mu_{S+2},...,\mu_{M_1}\to+\infty\), respectively.
\end{theorem}  
\begin{proof}
	See \cite{Ellipse}.
\end{proof}
Choose \(\mathbf{F}\) as the first \(N-M_2\) rows of \(\mathbf{I}_N\) and construct \(\mathbf{H}_3\) as (\ref{FR}). Denote \(\mu\) as the GSV of \(\{\mathbf{H}_1,\mathbf{H}_3\}\), then \(\mu\) equals to the eigenvalues of \(\mathbf{L} = \mathbf{H}_1(\mathbf{H}_3^H\mathbf{H}_3)^{-1}\mathbf{H}_1^H\), which can be derived by using the similar approach in the case of \(M_2>N\). Denote \(n = M_1+N\), construct
\begin{equation}
	\mathbf{J} = 
	\begin{pmatrix}
		\mathbf{O} &\mathbf{X}_3^H&\mathbf{O}&\mathbf{I}_{n}\\
		\mathbf{X}_3 &\mathbf{I}_{n} &\mathbf{O} &\mathbf{O}\\
		\mathbf{O} &\mathbf{O} &\mathbf{X}_1&-\mathbf{I}_{n}\\
		\mathbf{I}_{n} &\mathbf{O} &-\mathbf{I}_{n}&\mathbf{O}
	\end{pmatrix}.,
\end{equation}
where
\begin{equation}
	\mathbf{X}_1 = 
	\begin{pmatrix}
		\mathbf{O} &\mathbf{H}_1\\
		\mathbf{H}_1 &\mathbf{O}
	\end{pmatrix},
	\mathbf{X}_3 = 
	\begin{pmatrix}
		\mathbf{O}_{M_1\times M_1} &\mathbf{O}\\
		\mathbf{O} &\mathbf{H}_3
	\end{pmatrix}.
\end{equation}
Now \(\mathcal{G}^{\mathcal{D}_{M_1}}_{\mathbf{L}}(\mathbf{z\mathbf{I}_{M_1}})\) can be obtained by deriving \(\mathcal{G}^{\mathcal{D}_{4n}}_{\mathbf{J}}(\mathbf{C})\). \(\mathbf{J}\) can be expressed as the sum of a deterministic matrix and a random hermitian matrix as \(\mathbf{J} = \bar{\mathbf{J}}+\widetilde{\mathbf{J}}\), where the matrices \(\bar{\mathbf{J}}\) and \(\widetilde{\mathbf{J}}\) can be expressed as follows:
\begin{equation}
	\begin{aligned}
		\begin{small}
			\bar{\mathbf{J}} = \begin{pmatrix}
				\mathbf{O} &\bar{\mathbf{X}}_3^H&\mathbf{O}&\mathbf{I}_{n}\\
				\bar{\mathbf{X}}_3 &\mathbf{I}_{n} &\mathbf{O} &\mathbf{O}\\
				\mathbf{O} &\mathbf{O} &\bar{\mathbf{X}}_1&-\mathbf{I}_{n}\\
				\mathbf{I}_{n} &\mathbf{O} &-\mathbf{I}_{n}&\mathbf{O}
			\end{pmatrix},
			\widetilde{\mathbf{J}} = \begin{pmatrix}
				\mathbf{O} &\widetilde{\mathbf{X}}_3^H&\mathbf{O}&\mathbf{O}\\
				\widetilde{\mathbf{X}}_3 &\mathbf{O} &\mathbf{O} &\mathbf{O}\\
				\mathbf{O} &\mathbf{O} &\widetilde{\mathbf{X}}_1&\mathbf{O}\\
				\mathbf{O} &\mathbf{O} &\mathbf{O}&\mathbf{O}
			\end{pmatrix},
		\end{small}
	\end{aligned}
\end{equation}
\begin{equation}
	\bar{\mathbf{X}}_1 = 
	\begin{pmatrix}
		\mathbf{O}&\bar{\mathbf{H}}_1\\
		\bar{\mathbf{H}}_1^H&\mathbf{O}
	\end{pmatrix},
	\widetilde{\mathbf{X}}_1 = 
	\begin{pmatrix}
		\mathbf{O}&\widetilde{\mathbf{H}}_1\\
		\widetilde{\mathbf{H}}_1^H&\mathbf{O}
	\end{pmatrix},
\end{equation}
\begin{equation}\label{X3}
	\bar{\mathbf{X}}_3 = 
	\begin{pmatrix}
		\mathbf{O}_{M_1\times M_1} &\mathbf{O}\\
		\mathbf{O} &\bar{\mathbf{H}}_3
	\end{pmatrix},
	\widetilde{\mathbf{X}}_3 = 
	\begin{pmatrix}
		\mathbf{O}_{M_1\times M_1} &\mathbf{O}\\
		\mathbf{O} &\widetilde{\mathbf{H}}_3
	\end{pmatrix}.
\end{equation}
\begin{equation}
	\bar{\mathbf{H}}_3 = 
	\begin{pmatrix}
		\bar{\mathbf{H}}_2\\
		\epsilon\mathbf{F}
	\end{pmatrix},
	\widetilde{\mathbf{H}}_3 = 
	\begin{pmatrix}
		\widetilde{\mathbf{H}}_2\\
		\mathbf{O}
	\end{pmatrix}.
\end{equation}

Then the free deterministic equivalent of \(\mathbf{J}\) can be constructed as follows. Define \(\widetilde{\boldsymbol{\mathcal{H}}}_i\) as the same matrix in the case of \(M_2\geq N\). Then by replacing all \(\mathbf{H}_1\) with \(\boldsymbol{\mathcal{H}}_1 = \bar{\mathbf{H}}_1 + \widetilde{\boldsymbol{\mathcal{H}}}_1\) and \(\mathbf{H}_3\) with
\begin{equation}
	\boldsymbol{\mathcal{H}}_3 = 
	\begin{pmatrix}
		\bar{\mathbf{H}}_2\\
		\epsilon\mathbf{F}
	\end{pmatrix} + 
	\begin{pmatrix}
		\widetilde{\boldsymbol{\mathcal{H}}}_2\\
		\mathbf{O}
	\end{pmatrix},
\end{equation}
in \(\mathbf{J}\), \(\boldsymbol{\mathcal{J}}\) can be constructed as the free deterministic equivalent of \(\mathbf{J}\). \(\boldsymbol{\mathcal{J}}\) can be written as \(\boldsymbol{\mathcal{J}} = \bar{\mathbf{J}} + \widetilde{\boldsymbol{\mathcal{J}}}\), where \(\widetilde{\boldsymbol{\mathcal{J}}}\) can be constructed by replacing \(\widetilde{{\mathbf{H}}}_i\) with \(\widetilde{\boldsymbol{\mathcal{H}}}_i\) in \(\widetilde{\mathbf{J}}\). \(\boldsymbol{\mathcal{J}}\) and \(\mathbf{J}\) satisfy (\ref{JF}). The approach to derive \(\mathcal{G}^{\mathcal{D}_{4n}}_{\boldsymbol{\mathcal{J}}}(z\mathbf{I}_n)\) is presented as follows:

\begin{theorem}\label{IT2}
	Devide \(\mathcal{G}^{\mathcal{D}_{4n}}_{\boldsymbol{\mathcal{J}}}(\mathbf{C})\) into several block matrices as follows:
	
	Denote \(\mathcal{G}^{\mathcal{D}_{4n}}_{\boldsymbol{\mathcal{J}}}(\mathbf{C}) = \text{blkdiag}\{\mathbf{E}_1,\mathbf{E}_2,\mathbf{E}_3,\mathbf{E}_4\}\), and \(\mathbf{E}_i = \text{blkdiag}\{\mathbf{E}_{i,1},\mathbf{E}_{i,2}\}\), where \(\mathbf{E}_{1,1}\in\mathbb{C}^{M_1\times M_1}\), \(\mathbf{E}_{1,2}\in\mathbb{C}^{N\times N}\), \(\mathbf{E}_{2,1}\in\mathbb{C}^{M_1\times M_1}\), \(\mathbf{E}_{2,2}\in\mathbb{C}^{N\times N}\), \(\mathbf{E}_{3,1}\in\mathbb{C}^{M_1\times M_1}\), and \(\mathbf{E}_{3,2}\in\mathbb{C}^{N\times N}\) are diagonal matrices. Then \(\mathbf{E}_{i,j}\) satisfy the following equations:
	\begin{equation}\label{E112}
		\begin{aligned}
			\mathbf{E}_{1,1}  =  \text{diag}\{((z-\text{Tr}\{\mathbf{E}_{1,2}\})\mathbf{I}_{M_1}
			&-\bar{\mathbf{H}}_1\mathbf{A}_1^{-1}\bar{\mathbf{H}}_1^H)^{-1}\},
		\end{aligned}
	\end{equation}
	\begin{equation}\label{E122}
		\begin{aligned}
			\mathbf{E}_{1,2} =  \text{diag}\{(\mathbf{A}_1-(z-\text{Tr}\{\mathbf{E}_{1,2}\})^{-1}\bar{\mathbf{H}}_1^H\bar{\mathbf{H}}_1)^{-1}\},
		\end{aligned}
	\end{equation}
	\begin{equation}\label{E222}
		\begin{aligned}
			\mathbf{E}_{2,2}^{(1)} &= \text{diag}
			\{
			(\bar{\mathbf{H}}_2\mathbf{A}_2^{-1}\bar{\mathbf{H}}_2^H-(1+\text{Tr}\{\mathbf{E}_{1,2}\})\mathbf{I}_{M_2}\\
			-&\epsilon^2\bar{\mathbf{H}}_2\mathbf{A}_2^{-1}\mathbf{F}^H(\epsilon^2\mathbf{F}\mathbf{A}_2^{-1}\mathbf{F}^H-\mathbf{I}_{N-M_2})^{-1}\mathbf{F}\mathbf{A}_2^{-1}\bar{\mathbf{H}}_2^H)^{-1}
			\},
		\end{aligned}
	\end{equation}
	where
	\begin{equation}
		\begin{aligned}\label{AA1}
			\mathbf{A}_{1} = (1+\text{Tr}\{\mathbf{E}_{1,2}\})^{-1}\bar{\mathbf{H}}_2^H\bar{\mathbf{H}}_2&+\epsilon^2\mathbf{F}^H\mathbf{F}\\
			&-(\text{Tr}\{\mathbf{E}_{2,2}^{(1)}\}+\text{Tr}\{\mathbf{E}_{1,1}\})\mathbf{I}_N,
		\end{aligned}
	\end{equation}
	\begin{equation}\label{A22}
		\mathbf{A}_{2} = (z-\text{Tr}\{\mathbf{E}_{1,2}\})^{-1}\bar{\mathbf{H}}_1^H\bar{\mathbf{H}}_1+(\text{Tr}\{\mathbf{E}_{2,2}^{(1)}\}+\text{Tr}\{\mathbf{E}_{1,1}\})\mathbf{I}_N,
	\end{equation}
	\begin{equation}
		\mathbf{F} = 
		\begin{pmatrix}
			\mathbf{I}_{N-M_2},\mathbf{O}_{N-M_2\times M_2},
		\end{pmatrix}
	\end{equation}
	\(\mathbf{E}_{2,2}^{(1)}\) is the left-up \(M_2\times M_2\) block of \(\mathbf{E}_{2,2}\). 
	
\end{theorem}
\begin{proof}
	See Appendix \ref{IT2A}.
\end{proof}
Similarly, the numerical results of \(\mathbf{E}_{i,j}\) with a certain input \(z_0\) can be obtained with an iterative approach. Then \(G_{\mathbf{L}}(z_0)\) can be obtained by using (\ref{GL}). To obtain \(G_{\omega}(z_0)\), the following theorem is introduced:
\begin{theorem}\label{Gw2}
	When \(\epsilon\to 0\), \(G_{\omega}(z_0)\) converges to the following result:
	\begin{equation}\label{GwE2}
		G_{\omega}(z_0) \to \frac{M_1}{S}G_{\mathbf{L}}(z_0).
	\end{equation}
\end{theorem}
\begin{proof}
	See Appendix \ref{Gw2A}.
\end{proof}

Now \(G_{\omega}(z)\) can be obtained in both cases, and can be used to derive \(\bar{R}_i\), which will be discussed in the next section.

\subsection{The Results of Average Rates}
In this subsection, the relationship between \(\bar{R}_i\) and \(G_{\omega}(z)\) is studied. This requires the following theorem.
\begin{theorem}\label{rate}
	\(\bar{R}_i\) can be computed with \(G_{\omega}(z)\) as follows:
	\begin{equation}\label{R1}
		\bar{R}_1 = \log\left(1+\frac{l_1\rho}{td_1^{\tau}}\right) +  \int_{0}^{\frac{l_1\rho}{td_1^{\tau}}}\left(\frac{1}{(y+1)^2}G_{\omega}(-(y+1)^{-1})\right)dy,
	\end{equation}
	\begin{equation}\label{R2}
		\bar{R}_2 = \int_{0}^{\frac{\rho}{td_2^{\tau}}}\left(-G_{\omega}(-(1+y))+l_1G_{\omega}(-(1+l_1y))\right)dy,
	\end{equation}
\end{theorem}
\begin{proof}
	See Appendix \ref{rateA}. 
\end{proof}

By applying this theorem, the approximated numerical values of \(\bar{R}_i\) can be obtained by computing the numerical integrations in (\ref{R1}) and (\ref{R2}).

\section{The Special Case of Rayleigh Fading}
In \cite{Chen}, the average rates of GSVD-base MIMO-NOMA communications systems with Rayleigh fading were provided. However, the expressions are not in closed-forme. Besides, the numbers of antennas are restricted as \(M_1=M_2\). Motivated by this, the close-formed expressions of \(\bar{R}_i\) with Rayleigh fading are analyzed and derived based on the results presented in the previous section.

When the channel is modeled as Rayleigh fading, i.e., \(\bar{\mathbf{H}}_i = \mathbf{O}\), for the case of \(M_2\geq N\), (\ref{E11}) - (\ref{A2}) can be rewritten as follows:
\begin{equation}\label{RE11}
	\begin{aligned}
		\text{Tr}\{\mathbf{E}_{1,1}\}  &=  \text{Tr}\{(z-\text{Tr}\{\mathbf{E}_{1,2}\})^{-1}\mathbf{I}_{M_1}
		\} \\
		&= M_1(z-\text{Tr}\{\mathbf{E}_{1,2}\})^{-1},
	\end{aligned}
\end{equation}
\begin{equation}\label{RE12}
	\begin{aligned}
		\text{Tr}\{\mathbf{E}_{1,2}\} &=  -\text{Tr}\{(\text{Tr}\{\mathbf{E}_{2,2}\}+\text{Tr}\{\mathbf{E}_{1,1}\})^{-1}\mathbf{I}_N\}\\
		&=-N(\text{Tr}\{\mathbf{E}_{2,2}\}+\text{Tr}\{\mathbf{E}_{1,1}\})^{-1},
	\end{aligned}
\end{equation}
\begin{equation}\label{RE22}
	\begin{aligned}
		\text{Tr}\{\mathbf{E}_{2,2}\} &= -\text{Tr}
		\{
		(1+\text{Tr}\{\mathbf{E}_{1,2}\})^{-1}\mathbf{I}_{M_2}
		\}\\
		&=-M_2(1+\text{Tr}\{\mathbf{E}_{1,2}\})^{-1}.
	\end{aligned}
\end{equation}

And for the case of \(M_2< N <M_1+M_2\). (\ref{E112}) - (\ref{A22}) can be rewritten as follows:
\begin{equation}\label{RE112}
	\begin{aligned}
		\text{Tr}\{\mathbf{E}_{1,1}\} =  M_1(z-\text{Tr}\{\mathbf{E}_{1,2}\})^{-1},
	\end{aligned}
\end{equation}
\begin{equation}\label{RE122}
	\begin{aligned}
		\text{Tr}\{\mathbf{E}_{1,2}\} &=  -\text{Tr}\{(\text{Tr}\{\mathbf{E}_{2,2}^{(1)}\}+\text{Tr}\{\mathbf{E}_{1,1}\}-\epsilon^2\mathbf{F}^H\mathbf{F})^{-1}\mathbf{I}_N\}\\
		&
		\begin{aligned}
			&=-M_2(\text{Tr}\{\mathbf{E}_{2,2}^{(1)}\}+\text{Tr}\{\mathbf{E}_{1,1}\})^{-1}\\
			&-(N-M_2)(\text{Tr}\{\mathbf{E}_{2,2}^{(1)}\}+\text{Tr}\{\mathbf{E}_{1,1}\}-\epsilon^2)^{-1},
		\end{aligned}
	\end{aligned}
\end{equation}
\begin{equation}\label{RE222}
	\begin{aligned}
		\text{Tr}\{\mathbf{E}_{2,2}^{(1)}\} 
		=-M_2(1+\text{Tr}\{\mathbf{E}_{1,2}\})^{-1}.
	\end{aligned}
\end{equation}

Note when \(\epsilon\to 0\), these two iterative processes are equivalent. Thus the two cases \(M_2\geq N\) and \(M_2<N<M_1+M_2\) can be unified. Now the iterative process is simplified as a ternary quadratic equation set, which can be easily solved. 

\begin{theorem}\label{G}
	When \(\bar{\mathbf{H}}_i = \mathbf{O}\), the closed-form asymptotic results of \(G_{\mathbf{L}}(z)\) can be derived as follows:
	\begin{equation}\label{RGw}
		\begin{aligned}
			G_{\mathbf{L}}&(z) \approx \\
			&\left\{
			\begin{aligned}
				&\frac{M_1-N+2M_1z+M_2z-Nz+\sqrt{\Delta(z)}}{2M_1z(z+1)},z\neq -1,\\
				&-\frac{M_1+M_2-N}{M_1+M_2},z = -1,
			\end{aligned}
			\right.
		\end{aligned}
	\end{equation}
	where
	\begin{equation}\label{Deltaz}
		\begin{aligned}
			\Delta(z) = (M_1-N)^2-2Qz+(M_2-N)^2z^2,
		\end{aligned}
	\end{equation}
	\begin{equation}
		Q = NM_1+NM_2+M_1M_2-N^2.
	\end{equation}
\end{theorem}
\begin{proof}
	By solving (\ref{RE11}) - (\ref{RE22}), (\ref{RGw}) can be obtained.
\end{proof}

Define \(I(a,b) = \int_{a}^{b} G_{\mathbf{L}}(z)dz\). The close-formed expression of \(I(a,b)\) can be derived by taking the definite integral of (\ref{RGw}) as follows:
\begin{equation}\label{Iab}
	I(a,b) = \frac{M_1-N}{2M_1}\log\frac{b}{a}+\frac{1}{2M_1}(I_1(a,b)+I_2(a,b)+I_3(a,b)),
\end{equation} 
where
\begin{equation}
		\begin{aligned}
			&I_1(a,b) = -|M_1-N|\times\\
			&\log\left(\frac{a}{b}\frac{(M_1-N)^2-Qb+|M_1-N|\sqrt{\Delta(b)}}{(M_1-N)^2-Qa+|M_1-N|\sqrt{\Delta(a)}}\right),
		\end{aligned}
\end{equation}
\begin{equation}
	\begin{aligned}
		&I_2(a,b) = -|M_2-N|\times\\
		&\log\left(\frac{Q-(M_2-N)^2b+|M_2-N|\sqrt{\Delta(b)}}{Q-(M_2-N)^2a+|M_2-N|\sqrt{\Delta(a)}}\right),
	\end{aligned}
\end{equation}
\begin{equation}
	\begin{aligned}
		&I_3(a,b) = (M_1+M_2)\times\\
		&\log\left(\frac{Q_1-Q_2b+(M_1+M_2)\sqrt{\Delta(b)}}{Q_1-Q_2a+(M_1+M_2)\sqrt{\Delta(a)}}\right),
	\end{aligned}
\end{equation}
\begin{equation}
	\begin{aligned}
		&Q_1 = M_1^2+M_1M_2+NM_2-NM_1,\\
		&Q_2 = M_2^2+M_1M_2+NM_1-NM_2.
	\end{aligned}
\end{equation}
\(\Delta(a)\) and \(\Delta(b)\) are defined by substituting \(z\) with \(a\) and \(b\) in (\ref{Deltaz}), respectively. The integration only involves basic mathematical knowledge but is tedious, so the detailed process is not concluded in this paper. Then closed-form expressions of \(\bar{R}_i\) can be obtained as follows:

\begin{theorem}\label{Rrate}
	When \(\bar{\mathbf{H}}_i = \mathbf{O}\), the asymptotic result of \(\bar{R}_i\) can be expressed as follows:
	\begin{itemize}
		\item When \(M_2\geq N\).
		\begin{equation}\label{RR1}
			\begin{small}
				\begin{aligned}
					\bar{R}_1 \approx \frac{M_1}{S}\log\left(1+\frac{l_1\rho}{td_1^{\tau}}\right)\!+\!\frac{M_1}{S}I(-1,-\frac{td_1^{\tau}}{l_1\rho+\!td_1^{\tau}}),
				\end{aligned}
			\end{small}
		\end{equation}
		\begin{equation}\label{RR2}
			\begin{small}
				\begin{aligned}
					\bar{R}_2 \approx
					\frac{M_1}{S}I(-1-\frac{\rho l_1}{d_2^{\tau}},-1-\frac{\rho}{d_2^{\tau}})\!-\!\frac{M_1-S}{S}\ln\frac{d_2^{\tau}+\rho}{d_2^{\tau}+\rho l_1}.
				\end{aligned}
			\end{small}
		\end{equation}
	\item  When \(M_2<N<M_1+M_2\)
	\begin{equation}\label{RR12}
		\begin{small}
			\begin{aligned}
				\bar{R}_1 \approx \log\left(1+\frac{l_1\rho}{td_1^{\tau}}\right)+\frac{M_1}{S}I(-1,-\frac{td_1^{\tau}}{l_1\rho+\!td_1^{\tau}}),
			\end{aligned}
		\end{small}
	\end{equation}
	\begin{equation}\label{RR22}
		\begin{small}
			\begin{aligned}
				\bar{R}_2 \approx
				\frac{M_1}{S}I(-1-\frac{\rho l_1}{td_2^{\tau}},-1-\frac{\rho}{td_2^{\tau}}).
			\end{aligned}
		\end{small}
	\end{equation}
	\end{itemize}

\end{theorem}
\begin{proof}
	See Appendix \ref{RrateA}.
\end{proof}

In \cite{Chen}, the PDF of \(\omega\) is derived when \(M_1=M_2\). To compare the results, \(f_{\omega}(x)\) is derived from (\ref{RGw}) as follows:

\begin{theorem}\label{PDF}
	When \(\bar{\mathbf{H}}_i=\mathbf{O}\), the PDF of \(\omega\) can be expressed as follows:
	\begin{equation}\label{fw}
		f_{\omega}(x) \!=\! \left\{
		\begin{aligned}
			&\frac{\sqrt{2Qx\!-\!(M_2\!-\!N)^2x^2\!-\!(M_1\!-\!N)^2}}{2\pi Sx(x+1)}, x_1\!<\!x\!<\!x_2\\
			&0,\text{otherwise}
		\end{aligned}
		\right.,
	\end{equation}
	where
	\begin{equation}\label{x12}
		\begin{aligned}
			&x_1 = \frac{Q-2\sqrt{M_1M_2(NM_1+NM_2-N^2)}}{(M_2-N)^2},\\
			&x_2 = \frac{Q+2\sqrt{M_1M_2(NM_1+NM_2-N^2)}}{(M_2-N)^2}.
		\end{aligned}
	\end{equation} 
\end{theorem}
\begin{proof}
	See Appendix \ref{PDFA}.
\end{proof}

When \(M_1 = M_2\), \(f_{\omega}(x)\) can be verified as the same as the results in \cite{Chen}. This can also verify the analytical results in Theorem \ref{IT}, \ref{IT2} and \ref{G}.
    
\section{Simulation Results}

In this section, numerical results are presented to validate the proposed analytical results. All the numerical results are obtained from \(10^6\) simulation experiments with Matlab. The parameters of the channel are chosen as \(d_1 = 200\) m, \(d_2 = 2000\) m, \(\tau = 2,\). The power of the white noise is set as \(-20\) dBm.

Fig. \ref{fig_OMA} shows the numerical and analytical results of sum rates \(\bar{R}_1+\bar{R}_2\) achieved by NOMA. The two sets of antenna numbers are chosen as \((M_1,M_2,N)=(24,24,36)\) and \((M_1,M_2,N)=(36,48,60)\). The power allocation coefficient is set as \(l_1 = 0.9\). The deterministic matrices \(\bar{\mathbf{H}}_i\) are randomly generated. The analytical results are obtained from Theorem \ref{rate}. For comparison, the results of the traditional OMA scheme are also presented in Fig. \ref{fig_OMA}. As shown in this figure, the gaps between numerical and analytical results are negligible, which verifies the accuracy of the proposed approach. Besides, it can be observed from the figure that the NOMA scheme has higher sum rates than that of OMA scheme, which shows the superior performance of the NOMA scheme. In particular, the gaps between the rates of GSVD-NOMA and the other schemes increase as \(P\) increases, which demonstrates the great benefits of GSVD at high SNR.

Fig. \ref{fig_N1} presents the numerical and analytical results of rates with different numbers of antennas. The transmission power is set as \(P = 40\) dBm. The power allocation coefficient is set as \(l_1 = 0.05\). For the case of \(M_2<N<M_1+M_2\), the approximation parameter \(\epsilon\) is set as \(\epsilon = 10^{-5}\). The numbers of antennas \(M_1\), \(M_2\) and \(N\) increase proportionally, with \(\mu\geq 1\) as the coefficient. The proportions are chosen as \((M_1,M_2,N)=(\mu,\mu,\mu)\), \((M_1,M_2,N)=(\mu,2\mu,\mu)\), \((M_1,M_2,N)=(2\mu,2\mu,3\mu)\) and \((M_1,M_2,N)=(3\mu,4\mu,5\mu)\), respectively. The analytical results are obtained from Theorem \ref{rate}. To have the same channel conditions for different \(M_1,M_2,N\), the deterministic matrices \(\bar{\mathbf{H}}_i\) are set as all-one matrices. From Fig. \ref{fig_N1}, the curves of numerical and analytical show a similar performance, which verifies the analytical results. In addition, it is worth noting that the analytical results are still pretty close to the numerical results when the numbers of antennas are very small, for example, \((M_1,M_2,N)=(1,1,1)\). This indicates proposed approach in this paper shows a good performance even in small-scale MIMO communications, which demonstrates the generalization of the approach. 

\begin{figure}
	\centering
	\includegraphics[scale=0.6]{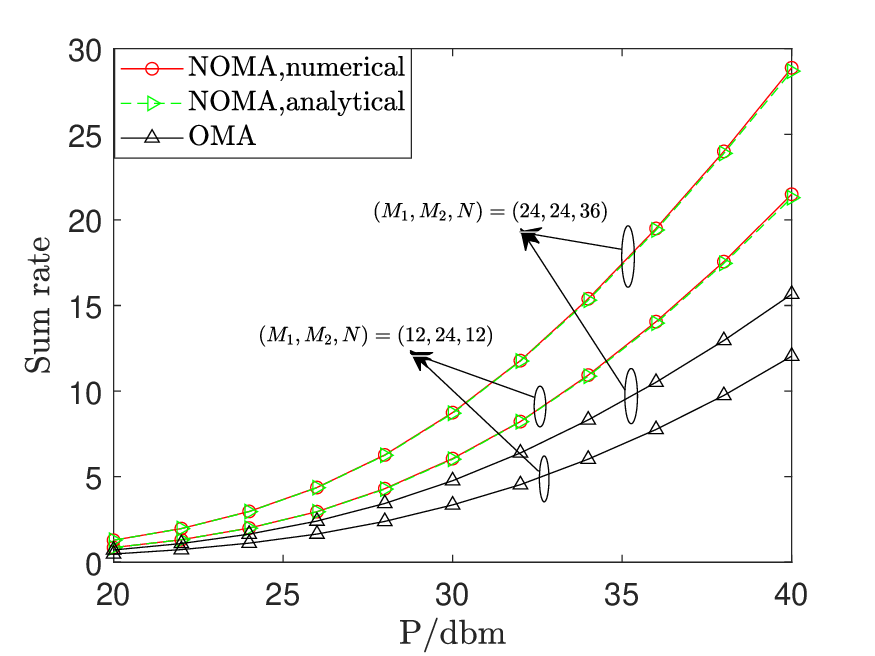}
	\caption{Numerical and analytical results of average sum rates \(\bar{R}_1+\bar{R}_2\) for the NOMA and OMA schemes with different SNR.}
	\label{fig_OMA}
\end{figure}

\begin{figure}
	\centering
	\subfigure[The case of \(M_2\geq N\).]
	{\includegraphics[scale=0.6]{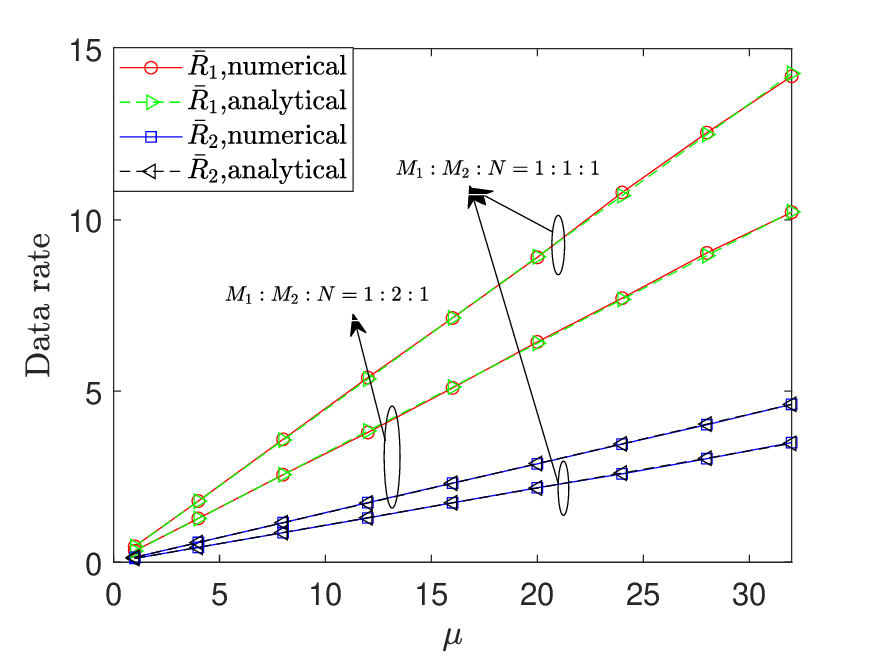}}
	\subfigure[The case of \(M_2< N <M_1+M_2\).]
	{\includegraphics[scale=0.6]{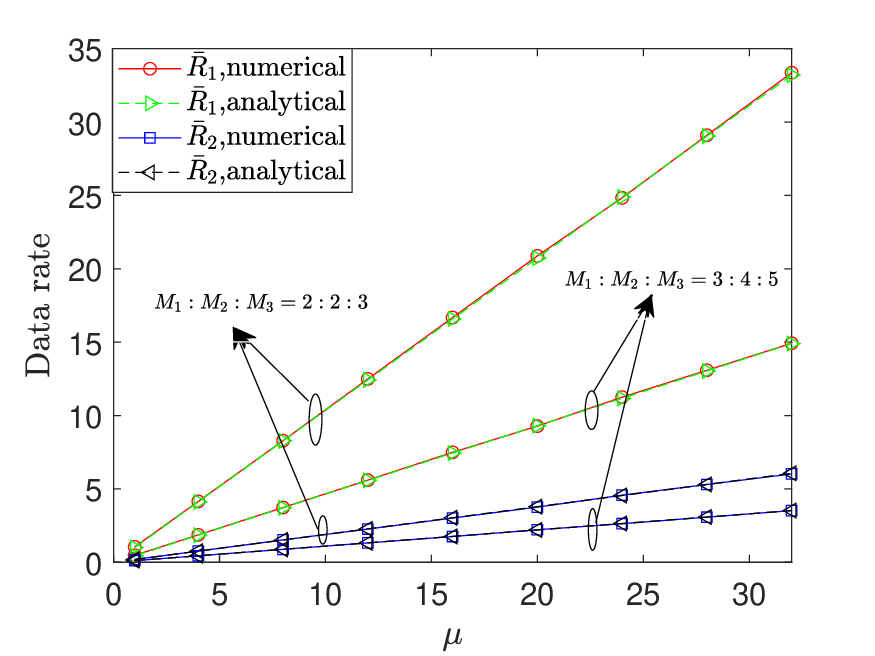}}
	\caption{Numerical and analytical results of average rates with different numbers of antennas. \(\mu\) represents the proportionality coefficient of the number of antennas.}
	\label{fig_N1}
\end{figure}

Fig. \ref{fig_04} presents the numerical results of the channel average data rate \(R_i = \frac{1}{S}\sum_{j=1}^{S}R_{i,j}\) as scatter points. The results of \(\frac{1}{S}\bar{R}_i\) are also presented as real lines for comparison. Channel proportions are chosen as \((M_1,M_2,N)=(\mu,2\mu,\mu)\) and \((M_1,M_2,N)=(2\mu,2\mu,3\mu)\), respectively. From Fig. \ref{fig_04}, \(R_i\) can be observed to converge to \(\frac{1}{S}\bar{R}_i\) when \(\mu\) increases. When \(\mu\geq 20\), the scatter points converge on the real lines, and \(\frac{1}{S}\bar{R}_i\) can be used to approximate the value of \(R_i\). This result manifests the asymptotic property of large-scale MIMO communication systems. 

\begin{figure}
	\centering
	\includegraphics[scale=0.6]{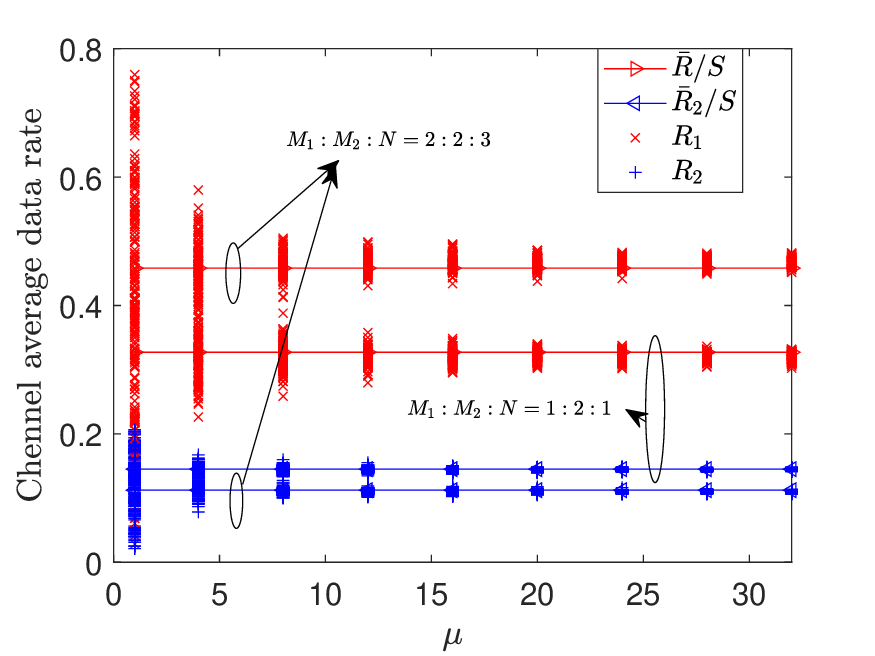}
	\caption{Numerical results of average channel data rate \(R_i = \frac{1}{S}\sum_{j=1}^{S}R_{i,j}\) and analytical results of average sum rates \(\frac{1}{S}\bar{R}_i\) with different numbers of antennas. \(\mu\) represents the proportionality coefficient of the number of antennas.}
	\label{fig_04}
\end{figure}

Fig. \ref{fig_ep} presents the numerical and analytical results of rates with different \(\epsilon\) when \(M_2<N<M_1+M_2\). The two sets of antennas are chosen as \((M_1,M_2,N)=(24,24,36)\) and \((M_1,M_2,N)=(36,48,60)\), respectively. \(\epsilon\) is chosen from \(10\) to \(0.1\). The deterministic matrices \(\bar{\mathbf{H}}_i\) are randomly generated. Since the only parameter that changes is the approximate parameter \(\epsilon\), the numerical results of average rates remain unchanged. It can be observed from Fig. \ref{fig_ep} that the accuracy of analytical results increases rapidly when \(\epsilon\) decreases. Specifically, when \(-\log\epsilon\geq 0.6\), i.e., \(\epsilon\leq 0.25\), the gap between numerical and analytical results can hardly be distinguished. This demonstrates the high accuracy of the proposed method in this paper.

\begin{figure}
	\centering
	\includegraphics[scale=0.6]{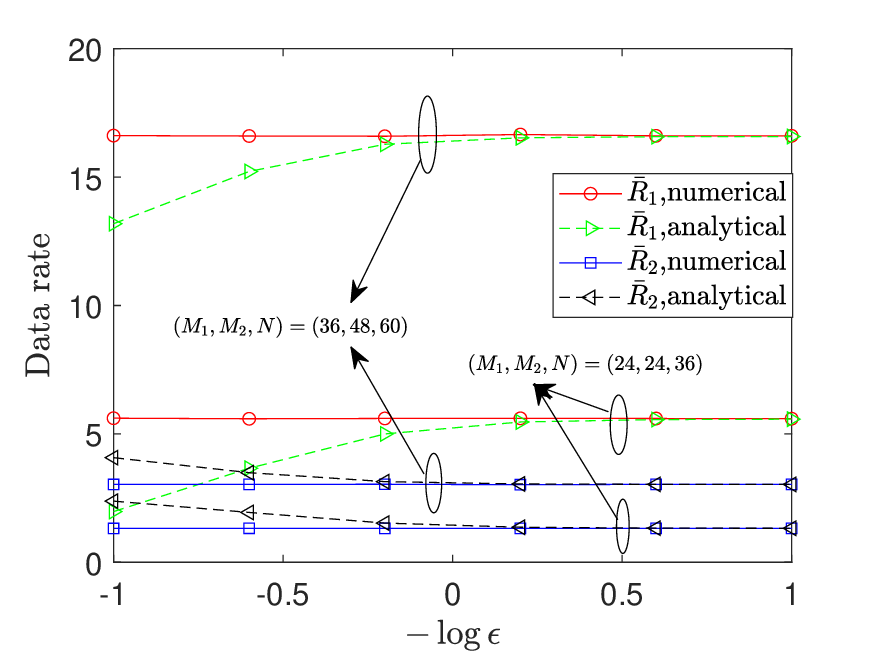}
	\caption{Numerical and analytical results of rates with different choose of \(\epsilon\) when \(M_2<N\).}
	\label{fig_ep}
\end{figure}	

Fig. \ref{fig_Ray} presents the numerical and analytical results of rates with different numbers of antennas when \(\bar{\mathbf{H}}_i=\mathbf{O}\), i.e., the channel is modeled as Rayleigh fading. The transmission power is set to be \(P = 40\) dBm. The power allocation coefficient is set as \(l_1 = 0.05\). The proportions are choosen as \((M_1,M_2,N)=(2\mu,2\mu,\mu)\) and \((M_1,M_2,N)=(3\mu,4\mu,5\mu)\), respectively. The analytical results are obtained from Theorem \ref{Rrate}. It can be seen from Fig. \ref{fig_Ray} that the differences between numerical and analytical results are almost invisible, which can verify the accuracy of the closed-form expressions in Theorem \ref{Rrate}. In addition, as is the case with the Rician distribution, the analytical results coincide with the numerical results well even \(\mu\) is very small, which shows the practical significance of the proposed approach.
 
Fig. \ref{fig_RayC} presents the closed-form analytical results in Theorem \ref{Rrate}, analytical results obtained with numerical integration in Theorem \ref{rate}, and the results in \cite{Chen}. Numerical results are also presented for comparison. Since the results in \cite{Chen} only considered the case of \(M_1=M_2\), the two sets of antennas are chosen as \((M_1,M_2,N)=(24,24,12)\) and \((M_1,M_2,N)=(48,48,60)\), respectively. The power allocation coefficient is set as \(l_1 = 0.05\). It can be seen that there is only a little difference between these results, which demonstrates the accuracy of these three methods.
 
\begin{figure}
	\centering
	\includegraphics[scale=0.6]{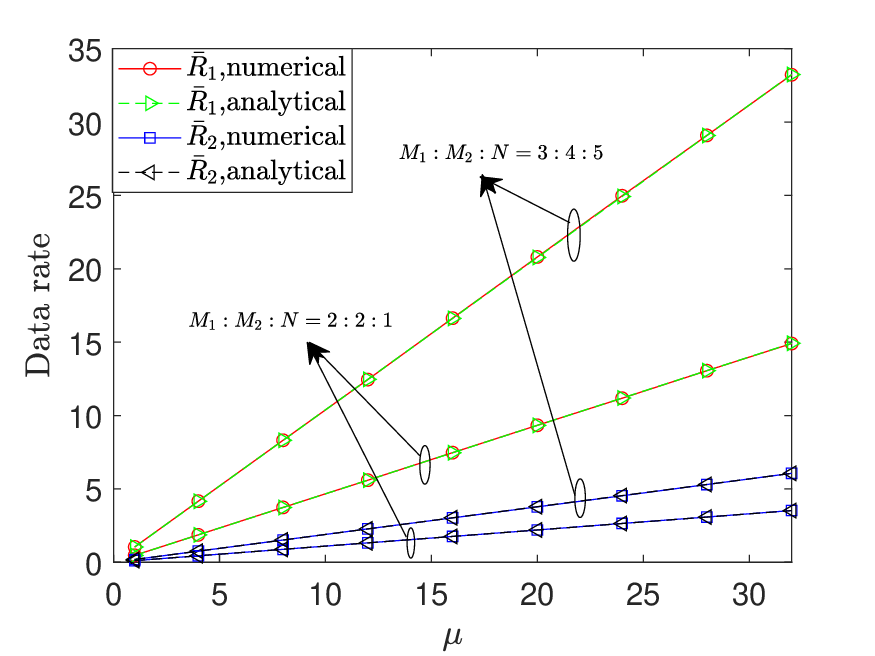}
	\caption{Numerical and analytical results of rates with different numbers of antennas when \(\bar{\mathbf{H}}_i=\mathbf{O}\). \(\mu\) represents the proportionality coefficient of the number of antennas.}
	\label{fig_Ray}
\end{figure}	

\begin{figure}
	\centering
	\includegraphics[scale=0.6]{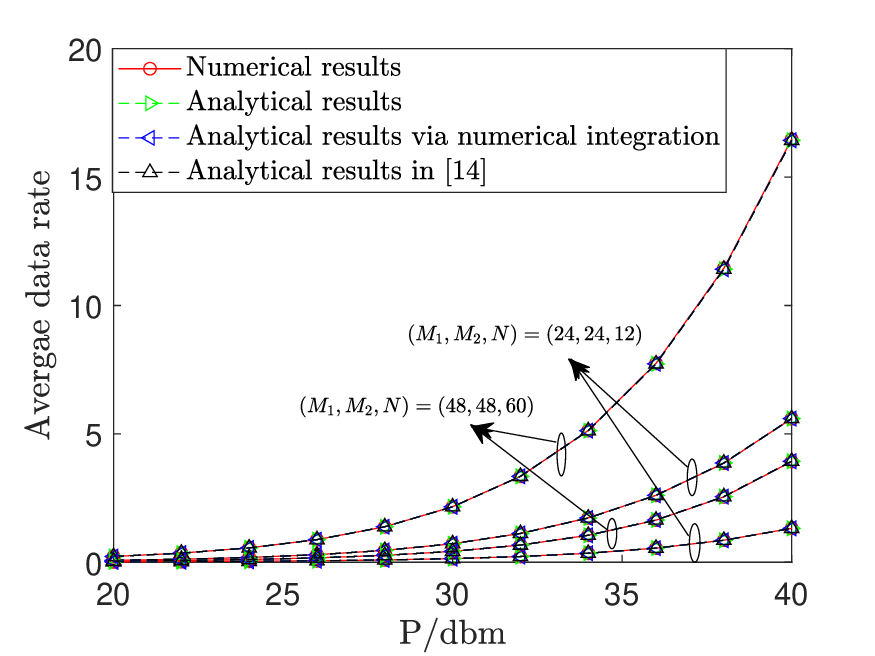}
	\caption{Numerical results and analytical results obtained by different methods when \(\bar{\mathbf{H}}_i=\mathbf{O}\) and \(M_1=M_2\).}
	\label{fig_RayC}
\end{figure}	

\section{Conclusion}
In this paper, a GSVD-based MIMO-NOMA communication system with Rician fading was considered. Based on the operator-valued free probability theory, the linearization trick and the deterministic equivalents method were exploited to obtain \(G_{\omega}(z_0)\), the Cauchy transform of GSVs of channel matrices. Then the average rates \(\bar{R}_i\) were obtained from \(G_{\omega}(z_0)\). In addition, the special case when \(\bar{\mathbf{H}}_i = \mathbf{O}\) was considered. The close-formed expressions of average rates are derived. Simulation results were provided to verify the accuracy of the analytical results.

\appendices

\section{Proof of Lemma \ref{ED}}\label{EDA}
Denote \(\mathbf{X} = \{\mathbf{x}_1,\mathbf{x}_2,...,\mathbf{x}_n\}\), \(\mathbf{Y} = \text{diag}\{y_1,...,y_n\}\), then \(E_{\mathcal{D}_{m}}\{\mathbf{X}\mathbf{Y}\mathbf{X}^H\}\) can be derived as
\begin{equation}
	\begin{aligned}
		E_{\mathcal{D}_{m}}\{\mathbf{X}\mathbf{Y}\mathbf{X}^H\} &= E_{\mathcal{D}_{m}}\{\sum_{i=1}^{n}y_i\mathbf{x}_i\mathbf{x}_i^H\}\\
		&=\sum_{i=1}^{n}y_iE_{\mathcal{D}_{m}}\{\mathbf{x}_i\mathbf{x}_i^H\} = \text{Tr}\{\mathbf{Y}\}\mathbf{I}_m.
	\end{aligned}
\end{equation}
The lemma is now proved.

\section{Proof of Theorem \ref{IT}}\label{ITA}

Since \(\mathbf{J} = \bar{\mathbf{J}} + \widetilde{\boldsymbol{\mathcal{J}}}\), \(\mathcal{G}^{\mathcal{D}_{4n}}_{\boldsymbol{\mathcal{J}}}(\mathbf{C})\) can be expressed with \(\bar{\mathbf{J}}\) and \(\widetilde{\boldsymbol{\mathcal{J}}}\) by applying subordination theorem as follows:
\begin{equation}\label{sub}
	\begin{aligned}
		\mathcal{G}^{\mathcal{D}_{4n}}_{\boldsymbol{\mathcal{J}}}(\mathbf{C}) &= 	\mathcal{G}^{\mathcal{D}_{4n}}_{\bar{\mathbf{J}}}(\mathbf{C}-\mathcal{R}^{\mathcal{D}_{4n}}_{\widetilde{\boldsymbol{\mathcal{J}}}}(\mathcal{G}^{\mathcal{D}_{4n}}_{\boldsymbol{\mathcal{J}}}(\mathbf{C})))\\
		&=E_{\mathcal{D}_{4n}}\{(\mathbf{C}-\mathcal{R}^{\mathcal{D}_{4n}}_{\widetilde{\boldsymbol{\mathcal{J}}}}(\mathcal{G}^{\mathcal{D}_{4n}}_{\boldsymbol{\mathcal{J}}}(\mathbf{C}))-\bar{\mathbf{J}})^{-1}\}.
	\end{aligned}
\end{equation}
where \(\mathcal{R}^{\mathcal{D}_{4n}}_{\widetilde{\boldsymbol{\mathcal{J}}}}()\) is the \(\mathcal{D}_{4n}\)-valued R-transform \cite{Lu}. Since \(\widetilde{\boldsymbol{\mathcal{J}}}\) is a hermitian matrix whose elements on and above the diagonal are freely independent, it is semicircular over \(\mathcal{D}_{12n}\) and free from the deterministic matrix \(\bar{G}\) \cite{Lu}. From \cite{Rt}, Th. 7.2, the \(\mathcal{D}_{4n}\)-valued R-transform of a semicircular variable can be written as follows:

\begin{equation}
	\mathcal{R}^{\mathcal{D}_{4n}}_{\widetilde{\boldsymbol{\mathcal{J}}}}(\mathbf{D}) = \mathcal{E}_{\mathcal{D}_{4n}}\{\widetilde{\boldsymbol{\mathcal{J}}}\mathbf{D}\widetilde{\boldsymbol{\mathcal{J}}}\}.
\end{equation}

Assume \(\mathcal{G}^{\mathcal{D}_{4n}}_{\boldsymbol{\mathcal{J}}}(\mathbf{C}) = \text{blkdiag}\{\mathbf{E}_1,\mathbf{E}_2,\mathbf{E}_3,\mathbf{E}_4\}\) and substitute it into \(\mathcal{R}^{\mathcal{D}_{4n}}_{\widetilde{\boldsymbol{\mathcal{J}}}}(\mathcal{G}^{\mathcal{D}_{4n}}_{\boldsymbol{\mathcal{J}}}(\mathbf{C}))\), we have
\begin{equation}\label{RT}
	\begin{aligned}
		\mathcal{R}^{\mathcal{D}_{4n}}_{\widetilde{\boldsymbol{\mathcal{J}}}}(\mathcal{G}^{\mathcal{D}_{4n}}_{\boldsymbol{\mathcal{J}}}(\mathbf{C})) &=  \text{blkdiag}\{\mathcal{E}_{\mathcal{D}_{n}}\{\widetilde{\boldsymbol{\mathcal{X}}}_2^H\mathbf{E}_2\widetilde{\boldsymbol{\mathcal{X}}}_2\},\\
		&\mathcal{E}_{\mathcal{D}_{n}}\{\widetilde{\boldsymbol{\mathcal{X}}}_2\mathbf{E}_1\widetilde{\boldsymbol{\mathcal{X}}}_2^H\},\mathcal{E}_{\mathcal{D}_{n}}\{\widetilde{\boldsymbol{\mathcal{X}}}_1\mathbf{E}_3\widetilde{\boldsymbol{\mathcal{X}}}_1\},\mathbf{O}\}.
	\end{aligned}
\end{equation}
To make it more concise, we denote
\begin{equation}
	\mathcal{R}^{\mathcal{D}_{4n}}_{\widetilde{\boldsymbol{\mathcal{J}}}}(\mathcal{G}^{\mathcal{D}_{4n}}_{\boldsymbol{\mathcal{J}}}(\mathbf{C})) =  \text{blkdiag}\{\mathbf{K}_1,
	\mathbf{K}_2,\mathbf{K}_3,\mathbf{O}\}.
\end{equation}

Then by substituting (\ref{RT}) into (\ref{sub}), the following equations can be derived:
\begin{equation}\label{E1}
	\mathbf{E}_1 = E_{\mathcal{D}_{n}}\{(\mathbf{B}-\mathbf{K}_1-\mathbf{K}_3-\bar{\mathbf{X}}_1+\bar{\mathbf{X}}_2^H(\mathbf{K}_2+\mathbf{I}_{n})^{-1}\bar{\mathbf{X}}_2)^{-1}\},
\end{equation}
\begin{equation}
	\mathbf{E}_2 = E_{\mathcal{D}_{n}}\{(-\mathbf{K}_2-\mathbf{I}_{n}-\bar{\mathbf{X}}_2(\mathbf{B}-\mathbf{K}_1-\mathbf{K}_3-\bar{\mathbf{X}}_1)^{-1}\bar{\mathbf{X}}_2^H)^{-1}\},
\end{equation}
\begin{equation}\label{E3}
	\mathbf{E}_3 = E_{\mathcal{D}_{n}}\{(\mathbf{B}-\mathbf{K}_1-\mathbf{K}_3-\bar{\mathbf{X}}_1+\bar{\mathbf{X}}_2^H(\mathbf{K}_2+\mathbf{I}_{n})^{-1}\bar{\mathbf{X}}_2)^{-1}\}.
\end{equation}

Then we are going to derive \(\mathbf{K}_i\). Denote the diagonal matrix \(\mathbf{E}_i = \text{blkdiag}\{\mathbf{E}_{i,1},\mathbf{E}_{i,2}\}\), where \(\mathbf{E}_{1,1}\in\mathbb{C}^{M_1\times M_1}\), \(\mathbf{E}_{1,2}\in\mathbb{C}^{N\times N}\), \(\mathbf{E}_{2,1}\in\mathbb{C}^{n-M_2\times n-M_2}\), \(\mathbf{E}_{2,2}\in\mathbb{C}^{M_2\times M_2}\), \(\mathbf{E}_{3,1}\in\mathbb{C}^{M_1\times M_1}\), \(\mathbf{E}_{3,2}\in\mathbb{C}^{N\times N}\). By substituting \(\mathbf{E}_2\) into \(\mathbf{K}_1 = E_{\mathcal{D}_{3n}}\{\widetilde{\boldsymbol{\mathcal{X}}}_2^H\mathbf{E}_2\widetilde{\boldsymbol{\mathcal{X}}}_2\}\), \(\mathbf{K}_1\) can be derived as follows:
\begin{equation}\label{k1}
	\mathbf{K}_1 = 
	\text{blkdiag}\{\mathbf{O},E_{\mathcal{D}_{N}}\{\widetilde{{\mathbf{H}}}_2^H\mathbf{E}_{2,2}\widetilde{{\mathbf{H}}}_2\}\}.
\end{equation}

From Lemma \ref{ED}, \(\mathbf{K}_1\) can be derived as follows:
\begin{equation}\label{k12}
	\begin{aligned}
		\mathbf{K}_1
		=\text{blkdiag}
		\{
		\mathbf{O},\text{Tr}\{\mathbf{E}_{2,2}\}\mathbf{I}_N
		\}.
	\end{aligned}
\end{equation}
Similarly, \(\mathbf{K}_2\) and \(\mathbf{K}_3\) can be derived as follows:
\begin{equation}
	\mathbf{K}_2 = \text{blkdiag}
	\{
	\mathbf{O},\text{Tr}\{\mathbf{E}_{1,2}\}\mathbf{I}_{M_2}
	\},
\end{equation}
\begin{equation}\label{K3}
	\mathbf{K}_3 = \text{blkdiag}
	\{
	\text{Tr}\{\mathbf{E}_{3,2}\}\mathbf{I}_{M_1},\text{Tr}\{\mathbf{E}_{3,1}\}\mathbf{I}_{N}
	\}.
\end{equation}

By substituting \(\mathbf{K}_2\) with the above results, \(\bar{\mathbf{X}}_2^H(\mathbf{K}_2+\mathbf{I}_{n})^{-1}\bar{\mathbf{X}}_2\) can be derived as follows:
\begin{equation}
	\bar{\mathbf{X}}_2^H(\mathbf{K}_2+\mathbf{I}_{n})^{-1}\bar{\mathbf{X}}_2 = \text{blkdiag}
	\{
	\mathbf{O},(1+\text{Tr}\{\mathbf{E}_{1,2}\})^{-1}\bar{\mathbf{H}}_2^H\bar{\mathbf{H}}_2
	\}.
\end{equation}

Denote \(\mathbf{A}_1\) as (\ref{A1}), note that \(\mathbf{E}_{1,j} = \mathbf{E}_{3,j}\), then \(\mathbf{E}_1\) and \(\mathbf{E}_3\) can be derived as follows:
\begin{equation}\label{E1E3}
	\begin{aligned}
		\mathbf{E}_1 = \mathbf{E}_3 = E_{\mathcal{D}_{n}}
		\{
		\begin{pmatrix}
			(z-\text{Tr}\{\mathbf{E}_{1,2}\})\mathbf{I}_{M_1}&-\bar{\mathbf{H}}_1\\
			-\bar{\mathbf{H}}_1^H&\mathbf{A}_1
		\end{pmatrix}^{-1}
		\}
		.
	\end{aligned}
\end{equation}
Extend (\ref{E1E3}) by applying the block matrix inverse formula in \cite{cookbook}, (\ref{E11}) and (\ref{E12}) can be derived.

Similarly, by substituting \(\mathbf{K}_1\) and \(\mathbf{K}_3\) with the results above, \(\bar{\mathbf{X}}_2(\mathbf{B}-\mathbf{K}_1-\mathbf{K}_3-\bar{\mathbf{X}}_1)^{-1}\bar{\mathbf{X}}_2^H\) can be derived as follows:
\begin{equation}
	\bar{\mathbf{X}}_2(\mathbf{B}-\mathbf{K}_1-\mathbf{K}_3-\bar{\mathbf{X}}_1)^{-1}\bar{\mathbf{X}}_2^H = \text{blkdiag}
	\{
	\mathbf{O},-\bar{\mathbf{H}}_2\mathbf{A}_2^{-1}\bar{\mathbf{H}}_2^H
	\}.
\end{equation}
Then \(\mathbf{E}_2\) can be derived as follows:
\begin{equation}
	\begin{aligned}
		&\mathbf{E}_2 \\&= E_{\mathcal{D}_{n}}
		\{
		\begin{pmatrix}
			-\mathbf{I}_{n\!-\!M_2}&\mathbf{O}\\
			\mathbf{O}&\bar{\mathbf{H}}_2\mathbf{A}_2^{-1}\bar{\mathbf{H}}_2^H\!-\!(1\!+\!\text{Tr}\{\mathbf{E}_{1,2}\})\mathbf{I}_{M_2}
		\end{pmatrix}^{-1}
		\},
	\end{aligned}
\end{equation}
which leads to (\ref{E22}), and the theorem is proved.

\section{Proof of Theorem \ref{Gw}}\label{GwA}
Denote \(d\) as a randomly chosen eigenvalue of \(\mathbf{L}\) (containing zero eigenvalues). Since \(\omega\) is randomly chosen from the \(S\) number of non-zero eigenvalues of \(\mathbf{L}\), \(G_{\mathbf{L}}(z_0)\) can be expressed as follows:
\begin{equation}
	\begin{aligned}
		G_{\mathbf{L}}(z_0) &= E\{(z_0-d)^{-1}\}\\
		&= \frac{S}{M_1}E\{(z_0-\omega)^{-1}\}+\frac{M_1-S}{M_1}E\{(z_0-0)^{-1}\}\\
		&= \frac{S}{M_1}G_{\omega}(z_0)+\frac{M_1-S}{M_1z_0},
	\end{aligned}
\end{equation}
which leads to (\ref{GwE}), and the theorem is proved.

\section{Proof of Theorem \ref{IT2}}\label{IT2A}
The approach is similar to the case \(M_2 > N\) in Appendix \ref{ITA}. Denote \(\mathcal{G}^{\mathcal{D}_{4n}}_{\boldsymbol{\mathcal{J}}}(\mathbf{C}) = \text{blkdiag}\{\mathbf{E}_1,\mathbf{E}_2,\mathbf{E}_3,\mathbf{E}_4\}\) and \(\mathcal{R}^{\mathcal{D}_{4n}}_{\widetilde{\boldsymbol{\mathcal{J}}}}(\mathcal{G}^{\mathcal{D}_{4n}}_{\boldsymbol{\mathcal{J}}}(\mathbf{C})) =  \text{blkdiag}\{\mathbf{K}_1,
\mathbf{K}_2,\mathbf{K}_3,\mathbf{O}\}\). Then by applying the subordination theorem, the following equations can be derived:
\begin{equation}\label{EE1}
	\mathbf{E}_1 = E_{\mathcal{D}_{n}}\{(\mathbf{B}-\mathbf{K}_1-\mathbf{K}_3-\bar{\mathbf{X}}_1+\bar{\mathbf{X}}_3^H(\mathbf{K}_2+\mathbf{I}_{n})^{-1}\bar{\mathbf{X}}_3)^{-1}\},
\end{equation}
\begin{equation}
	\mathbf{E}_2 = E_{\mathcal{D}_{n}}\{(-\mathbf{K}_2-\mathbf{I}_{n}-\bar{\mathbf{X}}_3(\mathbf{B}-\mathbf{K}_1-\mathbf{K}_3-\bar{\mathbf{X}}_1)^{-1}\bar{\mathbf{X}}_3^H)^{-1}\},
\end{equation}
\begin{equation}\label{EE3}
	\mathbf{E}_3 = E_{\mathcal{D}_{n}}\{(\mathbf{B}-\mathbf{K}_1-\mathbf{K}_3-\bar{\mathbf{X}}_1+\bar{\mathbf{X}}_3^H(\mathbf{K}_2+\mathbf{I}_{n})^{-1}\bar{\mathbf{X}}_3)^{-1}\}.
\end{equation}
\begin{equation}
	\begin{aligned}
		&\{\mathbf{K}_1,\mathbf{K}_2,\mathbf{K}_3\} = \\
		&\{\mathcal{E}_{\mathcal{D}_{n}}\{\widetilde{\boldsymbol{\mathcal{X}}}_3^H\mathbf{E}_2\widetilde{\boldsymbol{\mathcal{X}}}_3\},
		\mathcal{E}_{\mathcal{D}_{n}}\{\widetilde{\boldsymbol{\mathcal{X}}}_3\mathbf{E}_1\widetilde{\boldsymbol{\mathcal{X}}}_3^H\},\mathcal{E}_{\mathcal{D}_{n}}\{\widetilde{\boldsymbol{\mathcal{X}}}_1\mathbf{E}_3\widetilde{\boldsymbol{\mathcal{X}}}_1\}\}.
	\end{aligned}
\end{equation}

Since
\begin{equation}
	\begin{aligned}
		E_{\mathcal{D}_{N}}\{\widetilde{\mathbf{H}}_3^H\mathbf{E}_{2,2}\widetilde{\mathbf{H}}_3\} = E_{\mathcal{D}_{N}}\{\widetilde{\mathbf{H}}_2^H\mathbf{E}_{2,2}^{(1)}\widetilde{\mathbf{H}}_2\} = \text{Tr}\{\mathbf{E}_{2,2}^{(1)}\}\mathbf{I}_N,
	\end{aligned}
\end{equation}
where \(\mathbf{E}_{2,2}^{(1)}\) is the left-up \(M_2\times M_2\) block of \(\mathbf{E}_{2,2}\). Then \(\mathbf{K}_1\) can be derived as follows:
\begin{equation}
	\begin{aligned}
		\mathbf{K}_1
		=\text{blkdiag}
		\{
		\mathbf{O},\text{Tr}\{\mathbf{E}_{2,2}^{(1)}\}\mathbf{I}_N
		\}.
	\end{aligned}
\end{equation}
Similarly, \(\mathbf{K}_2\) can be derived as follows:
\begin{equation}
	\begin{aligned}
		\mathbf{K}_2
		=\text{blkdiag}
		\{
		\mathbf{O},\text{Tr}\{\mathbf{E}_{1,2}\}\mathbf{I}_{M_2},\mathbf{O}_{N-M_2}
		\}.
	\end{aligned}
\end{equation}
\(\mathbf{K}_3\) has the same expression with (\ref{K3}). By substituting \(\mathbf{K}_2\) with the results above, \(\bar{\mathbf{X}}_3^H(\mathbf{K}_2+\mathbf{I}_{n})^{-1}\bar{\mathbf{X}}_3\) can be derived as follows:
\begin{equation}
	\begin{aligned}
		&\bar{\mathbf{X}}_3^H(\mathbf{K}_2+\mathbf{I}_{n})^{-1}\bar{\mathbf{X}}_3 \\
		&=\text{blkdiag}
		\{
		\mathbf{O},(1+\text{Tr}\{\mathbf{E}_{1,2}\})^{-1}\bar{\mathbf{H}}_2^H\bar{\mathbf{H}}_2 + \epsilon^2\mathbf{J}^H\mathbf{J}
		\}.
	\end{aligned}
\end{equation}
Denote \(\mathbf{A}_1\) as (\ref{A1}), and \(\mathbf{E}_1\) and \(\mathbf{E}_3\) can be derived as the same expression as (\ref{E1E3}). Then (\ref{E112}) and (\ref{E122}) can be obtained.

Similarly, by substituting \(\mathbf{K}_1\) and \(\mathbf{K}_3\) with the results above, \(\bar{\mathbf{X}}_3(\mathbf{B}-\mathbf{K}_1-\mathbf{K}_3-\bar{\mathbf{X}}_1)^{-1}\bar{\mathbf{X}}_3^H\) can be derived as follows:
\begin{equation}
	\begin{aligned}
		\bar{\mathbf{X}}_3(&\mathbf{B}\!-\!\mathbf{K}_1\!-\!\mathbf{K}_3\!-\!\bar{\mathbf{X}}_1)^{-1}\bar{\mathbf{X}}_3^H = \text{blkdiag}
		\{
		\mathbf{O},-\bar{\mathbf{H}}_3\mathbf{A}_2^{-1}\bar{\mathbf{H}}_3^H
		\}\\
		&=
		\begin{pmatrix}
			\mathbf{O}&\mathbf{O}&\mathbf{O}\\
			\mathbf{O}&\bar{\mathbf{H}}_2\mathbf{A}_2^{-1}\bar{\mathbf{H}}_2^H&\epsilon\bar{\mathbf{H}}_2\mathbf{A}_2^{-1}\mathbf{J}^H\\
			\mathbf{O}&\epsilon\mathbf{J}\mathbf{A}_2^{-1}\bar{\mathbf{H}}_2^H&\epsilon^2\mathbf{J}\mathbf{A}_2^{-1}\bar{\mathbf{J}}^H
		\end{pmatrix}.
	\end{aligned}
\end{equation}
Then \(\mathbf{E}_2\) can be derived as follows:
\begin{equation}
	\begin{aligned}
		\mathbf{E}_2 = 
		\begin{pmatrix}
			-\mathbf{I}_{M_1}&\mathbf{O}\\
			\mathbf{O}&E_{\mathcal{D}_{n}}\{\mathbf{T}^{-1}\}
		\end{pmatrix}^{-1},
	\end{aligned}
\end{equation}
where
\begin{equation}\label{T}
	\begin{aligned}
		&\mathbf{T} = \\
		&\begin{pmatrix}
			\bar{\mathbf{H}}_2\mathbf{A}_2^{-1}\bar{\mathbf{H}}_2^H\!-\!(1\!+\!\text{Tr}\{\mathbf{E}_{1,2}\})\mathbf{I}_{M_2}&\epsilon\bar{\mathbf{H}}_2\mathbf{A}_2^{-1}\mathbf{J}^H\\
			\epsilon\mathbf{J}\mathbf{A}_2^{-1}\bar{\mathbf{H}}_2^H&\epsilon^2\mathbf{J}\mathbf{A}_2^{-1}\bar{\mathbf{J}}^H\!-\!\mathbf{I}_{N\!-\!M_2}
		\end{pmatrix}.
	\end{aligned}
\end{equation}
\(\mathbf{E}_{2,2}^{(1)}\) can be obtained by applying the block matrix inverse formula to (\ref{T}), and the theorem is proved.

\section{Proof of Theorem \ref{Gw2}}\label{Gw2A}
Denote \(d\) as a randomly chosen eigenvalue of \(\mathbf{L}\). Then \(G_{\mathbf{L}}(z_0)\) can be expressed as follows:
\begin{equation}
	\begin{aligned}
		G_{\mathbf{L}}(z_0) &= E\{(z_0-d)^{-1}\}\\
		&= \frac{S}{M_1}E\{(z_0-\omega')^{-1}\}+\frac{M_1-S}{M_1}E\{(z_0-e)^{-1}\},
	\end{aligned}
\end{equation}
where \(\omega'\) denotes randomly chosen element from \(\{\mu_1,\mu_2,...,\mu_S\}\), \(e\) denotes the randomly chosen element from \(\{\mu_{S+1},\mu_{S+2},...,\mu_{M_1}\}\). From Theorem \ref{ellipse}, when \(\epsilon\to 0\), \(\omega'\to\omega\) \(e\to+\infty\). Then the approximated value of \(G_{\mathbf{L}}(z_0)\) can be known as follows:
\begin{equation}
	G_{\mathbf{L}}(z_0) \to \frac{S}{M_1}E\{(z_0-\omega)^{-1}\} = \frac{S}{M_1}G_{\omega}(z_0),
\end{equation}
which leads to (\ref{GwE2}) can be derived, and the theorem is proved.

\section{Proof of Theorem \ref{rate}}\label{rateA}
Denote \(a = \frac{l_1\rho}{td_1^{\tau}}\), by differentiating (\ref{aR1}) with respect to \(a\), we have that
\begin{equation}
	\begin{aligned}
		\frac{d\bar{R}_1}{da} &= \int_{0}^{+\infty}\frac{x}{(a+1)x+1}f_{\omega}(x)dx\\
		&=\frac{1}{a+1}+\frac{1}{(a+1)^2}G_{\omega}(-(a+1)^{-1}).
	\end{aligned} 
\end{equation}
Since \(\bar{R}_1 = 0\) when \(a = 0\), \(\bar{R}_1\) can be expressed as (\ref{R1}). 

Denote \(b = \frac{\rho}{td_2^{\tau}}\), then by differentiating (\ref{aR2}) with respect to \(b\), we have that
\begin{equation}
	\begin{aligned}
		\frac{d\bar{R}_2}{db} =& \int_{0}^{+\infty}\frac{1}{x+1+b}\cdot\frac{l_2x+l_2}{x+1+l_1b}f_{\omega}(x)dx\\
		&-\frac{l_1l_2}{1-l_1}\int_{0}^{+\infty}\frac{1}{x+1+l_1b}f_{\omega}(x)dx\\
		=&-G_{\omega}(-(1+b))+l_1G_{\omega}(-(1+l_1b)).
	\end{aligned} 
\end{equation}
Since \(\bar{R}_2 = 0\) when \(b = 0\), \(\bar{R}_2\) can be expressed as (\ref{R2}). The theorem is proved.

\section{Proof of Theorem \ref{Rrate}}\label{RrateA}
When, \(M_2\geq N\), from Theorem \ref{Gw}, \(\int_a^b G_{\omega}(z)dz\) can be expressed by \(I(a,b) = \int_{a}^{b} G_{\mathbf{L}}(z)dz\) as follows:
\begin{equation}
	\int_a^b G_{\omega}(z)dz = \frac{M_1}{S}I(a,b)-\frac{M_1-S}{S}\ln \frac{b}{a}.
\end{equation}
Denote \(x = -\frac{1}{y+1}\), (\ref{R1}) can be rewritten as follows:
\begin{equation}
	\begin{aligned}
		\bar{R}_1 &= \log\left(1+\frac{l_1\rho}{td_1^{\tau}}\right) +  \int_{-1}^{-\frac{td_1^{\tau}}{l_1\rho+ td_1^{\tau}}}G_{\omega}(x)dx\\
		&=\log\left(1+\frac{l_1\rho}{td_1^{\tau}}\right)+\frac{M_1}{S}I(-1,-\frac{td_1^{\tau}}{l_1\rho+td_1^{\tau}})\\
		&-\frac{M_1-S}{S}\ln(\frac{td_1^{\tau}}{l_1\rho+ td_1^{\tau}}),
	\end{aligned}
\end{equation}
and then (\ref{RR1}) can be derived.

Denote \(y_1 = -(1+y)\), \(y_2 = -(1+l_1y)\), (\ref{R2}) can be rewritten as follows:
\begin{equation}
	\begin{aligned}
		\bar{R}_2 &= \int_{-1}^{-1-\frac{\rho}{td_2^{\tau}}}G_{\omega}(y_1)dy_1-\int_{-1}^{-1-\frac{\rho l_1}{td_2^{\tau}}}G_{\omega}(y_2)dy_2\\
		&=\frac{M_1}{S}I(-1,-1-\frac{\rho}{td_2^{\tau}})
		-\frac{M_1-S}{S}\ln(1+\frac{\rho}{td_2^{\tau}})\\
		&-\frac{M_1}{S}I(-1,-1-\frac{\rho l_1}{td_2^{\tau}})
		+\frac{M_1-S}{S}\ln(1+\frac{\rho l_1}{td_2^{\tau}}),
	\end{aligned}
\end{equation}
and then (\ref{RR2}) can be derived.

Similarly, when \(M_2<N<M_1+M_2\), from Theorem \ref{Gw2}, \(\int_a^b G_{\omega}(z)dz\) can be expressed as follows:
\begin{equation}
	\int_a^b G_{\omega}(z)dz = \frac{M_1}{S}I(a,b).
\end{equation}
Then (\ref{R1}) and (\ref{R2}) can be rewritten as follows:
\begin{equation}
	\bar{R}_1 = log\left(1+\frac{l_1\rho}{td_1^{\tau}}\right)+\frac{M_1}{S}I(-1,-\frac{td_1^{\tau}}{l_1\rho+td_1^{\tau}}),
\end{equation}
\begin{equation}
	\bar{R}_2 = \frac{M_1}{S}I(-1-\frac{\rho l_1}{td_2^{\tau}},-1-\frac{\rho}{td_2^{\tau}}),
\end{equation}
respectively.

The theorem is proved.

\vspace{-1em}

\section{Proof of Theorem \ref{PDF}}\label{PDFA}
When \(M_2\geq N\), from \cite{FPT}, the PDF of \(\omega\) can be derived from
\begin{equation}
	\begin{aligned}
		f_{\omega}(x) &= -\lim\limits_{y\to 0^+}\frac{1}{\pi}\text{Im}\{G_{\omega}(x+iy)\}\\
		&=-\lim\limits_{y\to 0^+}\frac{M_1}{\pi S}\text{Im}\{G_{\mathbf{L}}(x+iy)\}-\frac{M_1-S}{S(x+iy)},
	\end{aligned}
\end{equation}
where \(\text{Im}\{z\}\) represents the imaginary part of \(z\). It is obvious that
\begin{equation}
	\lim\limits_{y\to 0^+}\text{Im}\{\frac{M_1\!-\!N\!+\!2M_1z\!+\!M_2z\!-\!Nz}{(x+iy)(x+1+iy)}\!-\!\frac{M_1\!-\!S}{S(x\!+\!iy)}\} = 0.
\end{equation}

When \(z = x+iy\), \(x > 0\), \(\lim\limits_{y\to 0^+}\text{Im}\{\frac{\sqrt{\Delta(z)}}{(z)(z)}\}\) can be written as follows:
\begin{equation}
	\begin{aligned}
		\lim\limits_{y\to 0^+}\text{Im}\{\frac{\sqrt{\Delta(z)}}{z(z+1)}\} &= \lim\limits_{y\to 0^+}\text{Im}\{\frac{\sqrt{\Delta(x+iy)}}{(x+iy)(x+1+iy)}\}\\
		&=\lim\limits_{y\to 0^+}\text{Im}\{\frac{\sqrt{|\Delta(x+iy)|}e^{j\frac{\theta}{2}}}{|x+iy||x+1+iy|}\}\\
		&=\lim\limits_{y\to 0^+}\frac{\sqrt{|\Delta(x)|}}{x(x+1)}\sin(\frac{\theta}{2}),
	\end{aligned}
\end{equation}
where \(\theta\) denotes the angles of \(\Delta(x+iy)\). Extend \(\Delta(x+iy)\) as
\begin{equation}
	\begin{aligned}
		\Delta(x+iy) =& ((M_2-N)^2(x^2-y^2)+(M_2-N)^2-2Qx) \\
		&+ 2((M_1-N)^2x-Q)yi,
	\end{aligned}
\end{equation}
then \(\lim\limits_{y\to 0^+}\theta\) can be known as follows:
\begin{equation}
	\lim\limits_{y\to 0^+}\theta = \left\{
	\begin{aligned}
		&\pi,x_1<x<x_2\\
		&0,x < x_1\ \text{or}\ x > x_2
	\end{aligned}
	\right.,
\end{equation}
where \(x_1\) and \(x_2\) are roots of \((M_2-N)^2x^2+(M_1-N)^2-2Qx = 0\), which can be derived as (\ref{x12}). Then \(f_{\omega}(x)\) can be derived as follows:
\begin{equation}\label{fx}
	f_{\omega}(x) = \left\{
	\begin{aligned}
		&\frac{\sqrt{-\Delta(x)}}{{2\pi Sx(x+1)}}, x_1<x<x_2\\
		&0,\text{otherwise}
	\end{aligned}
	\right..
\end{equation}

When \(M_2<N<M_1+M_2\), \(f_{\omega}(x)\) can be written as follows:
\begin{equation}
	f_{\omega}(x)=-\lim\limits_{y\to 0^+}\frac{M_1}{\pi S}\text{Im}\{G_{\mathbf{L}}(x+iy)\},
\end{equation}
which is equivalent to the case of \(M_2\geq N\). Thus the expression of PDF is equal to (\ref{fx}). The theorem is now proved. 

\bibliographystyle{IEEEtran}
\bibliography{ref}

\end{document}